\documentclass[letterpaper,
twocolumn,
aps,nofootinbib,superscriptaddress,tightenlines,10pt,notitlepage,pra]{revtex4-1}
\pdfoutput=1

\usepackage{amscd,amsfonts,amsmath,amssymb,amsthm,dsfont,mathrsfs,mathtools,amsbsy}
\usepackage{color,graphicx,enumerate}
\usepackage{natbib}
\usepackage[
colorlinks=true, 
linkcolor=blue, 
citecolor=blue, 
urlcolor=blue, 
hyperindex,
breaklinks
]{hyperref}

\usepackage{color}

\usepackage[utf8]{inputenc}
\usepackage[T1]{fontenc}
\usepackage{lmodern}
\usepackage{centernot}

\usepackage{blindtext}

\usepackage{tikz}
\usetikzlibrary{arrows}

\theoremstyle{plain}
\newtheorem{theorem}{Theorem}
\newtheorem{lemma}{Lemma}

\theoremstyle{definition}

\providecommand{\ket}[1]{\lvert#1\rangle}
\providecommand{\bra}[1]{\langle#1\rvert}

\providecommand{\ketbra}[2]{\rvert#1\rangle\!\langle#2\lvert}
\providecommand{\braketinner}[3]{\langle#1|#2|#3\rangle}

\newcommand{\orb}{\operatorname{orb}}
\newcommand\walpha{a}
\newcommand\wbeta{b}

\def\locc{\xrightarrow{\textup{\tiny{LOCC}}}}
\newcommand{\cohat}[1]{\widehat{#1}}%{\operatorname{co}#1}%{#1^{\operatorname{co}}}
\newcommand{\co}[1]{\operatorname{co}#1}%{#1^{\operatorname{co}}}

\def\:{\colon}

\newcommand\D{\mathcal{D}}

\newcommand\U{\operatorname{U}}
\newcommand\Orth{\operatorname{O}}
\newcommand\LU{\operatorname{LU}}

\newcommand{\wer}{\operatorname{wer}}
\newcommand{\isot}{\operatorname{iso}}%{\textsc{iso}}%{\mathrm I}

\DeclareMathOperator{\Tr}{Tr}

\def\CC{\mathbb{C}}
\def\RR{\mathbb{R}}
\def\calT{\mathcal{T}}
\def\calG{\mathcal{G}}
\def\calH{\mathcal{H}}

\providecommand{\abs}[1]{\left\lvert#1\right\rvert}
\providecommand{\norm}[1]{\left\lVert#1\right\rVert}

\providecommand{\vecbf}[1]{\mathbf{#1}}

\renewcommand{\vec}[1]{{\boldsymbol{#1}}}

\begin{document}

\title{Entanglement monotones and transformations of symmetric bipartite states}
\date{\today}
\author{Mark W.\ Girard}
\email{mwgirard (at) ucalgary.ca}
\affiliation{Department of Mathematics and Statistics, University of Calgary}
\affiliation{Institute for Quantum Science and Technology, University of Calgary\\ 2500 University Dr NW, Calgary, AB T2N 1N4, Canada}
\author{Gilad Gour}
\affiliation{Department of Mathematics and Statistics, University of Calgary}
\affiliation{Institute for Quantum Science and Technology, University of Calgary\\ 2500 University Dr NW, Calgary, AB T2N 1N4, Canada}

\begin{abstract}

The primary goal in the study of entanglement as a resource theory is to find conditions that determine when one quantum state can or cannot be transformed into another via local operations and classical communication. This is typically done through entanglement monotones or conversion witnesses. Such quantities cannot be computed for arbitrary quantum states in general, but it is useful to consider classes of symmetric states for which closed-form expressions can be found. In this paper, we show how to compute the convex roof of any entanglement monotone for all Werner states. The convex roofs of the well-known Vidal monotones are computed for all isotropic states, and we show how this method can generalize to other entanglement measures and other types of symmetries as well. We also present necessary and sufficient conditions that determine when a pure bipartite state can be deterministically converted into a Werner state or an isotropic state. 

\end{abstract}

\maketitle

%~~~~~~~~~~~~~~~~~~~~~~~~~~~~~~~~~~~~~~~~~~~~~~~~~~~~~~~

\section{Introduction}

One of the main goals in quantum information theory has been to develop the theory of entanglement as a resource \cite{Vedral1997,Wootters1998a,Ishizaka2005,Plenio2007,Plenio2014}. The resource theory of entanglement arises naturally from the ``distant labs'' setting \cite{Horodecki2009}, in which the class of operations that can be performed is known as LOCC (local operations and classical communication). This is the class of local operations that can be implemented by separated parties acting locally upon their particles in a quantum manner, while coordinating their actions with the use of classical communication. 

The resourcefulness of entanglement is well known, and many famous quantum information processing tasks explicitly require the consumption of entanglement. It is therefore important to quantify the amount of entanglement in a state and to understand the structure of the states that is imposed by the restriction to LOCC. While the structure of entanglement in bipartite pure states has been completely characterized \cite{Nielsen1999}, a complete understanding of entanglement in arbitrary bipartite mixed states remains elusive. 

Entanglement is typically quantified via entanglement monotones---quantities that do not increase on average under local measurements and classical communication \cite{Vidal2000a}. Entanglement monotones on pure states can be obtained from functions of the vectors of Schmidt coefficients. Without loss of generality, we may write all bipartite pure states in Schmidt form as $\ket{\psi}=\sum_{i}\sqrt{\lambda_i}\ket{ii}$ for some Schmidt vector $\vec{\lambda}$. Every entanglement monotone corresponds to a function $f$ from the simplex of probability vectors to the real numbers, where $f$ is both symmetric and concave \cite{Vidal2000a}. Each such function defines an entanglement monotone on pure states $\ket{\psi}$, $E_f(\psi) = f(\vec{\lambda})$, by evaluating $f$ on the vector of Schmidt coefficients $\vec\lambda$ of $\ket{\psi}$, and every entanglement monotone can be obtained this way. Some well-known entanglement monotones include the entropy of entanglement $E(\psi) = H(\vec{\lambda})$ (where $H(\vec{\lambda})=-\sum_{i}\lambda_i\log\lambda_i$ is the Shannon entropy) and the Renyi $\alpha$ entropies of entanglement $E_\alpha(\psi) = H_\alpha(\vec{\lambda})$ (where $H_\alpha(\vec{\lambda})=\frac{1}{1-\alpha}\log\bigl(\sum_{i}\lambda_i^\alpha\bigr)$ is the Renyi $\alpha$ entropy of a probability distribution for $\alpha>0$ and $\alpha\neq 1$). Another class of important monotones for bipartite pure states was introduced by Vidal \cite{Vidal1999} and is defined as follows. Assuming $\ket{\psi}\in\CC^d\otimes\CC^d$ with $d\geq 2$, for each $k\in\{1,\dots,d-1\}$ we can define the monotones as
\begin{equation}\label{eq:vidalmonotones}
 E_k(\psi) = 1-\sum_{i=1}^k \lambda_i = \sum_{i=k+1}^d\lambda_i,
\end{equation}
where  the Schmidt coefficients of $\ket{\psi}$ are in decreasing order $\lambda_1\geq \cdots\geq\lambda_d$. For each $k$, $E_k(\psi)$ is the sum of the $d-k$ smallest Schmidt coefficients of $\ket{\psi}$. 

Entanglement monotones on pure states can be extended to arbitrary mixed states by a convex roof construction \cite{Vidal2000a,Uhlmann2010,Bucicovschi2011}. Given an entanglement monotone $E$ on pure states, its convex roof on mixed states is defined as
\[
 \widehat{E}(\rho) = \inf_{\{p_i,\ket{\psi_i}\}}\sum_{i}p_i E(\psi_i),
\]
where the infimum is taken over all pure state decompositions of $\rho=\sum_i p_i \ket{\psi_i}\bra{\psi_i}$. While there are many known entanglement monotones for bipartite pure states, evaluating the entanglement of arbitrary mixed states is in general not possible. In this paper, we show how to compute convex roof entanglement monotones on certain symmetric classes of entangled states. In particular, we compute the convex roofs of the Renyi entropies and the Vidal monotones on Werner and isotropic states \cite{Vollbrecht2001,Terhal2000}. 

Symmetry plays a very important role in many quantum information tasks. Restricting our attention to highly symmetric states not only simplifies many computations, but yields valuable information about the structure of bipartite entanglement. There is strong evidence that certain symmetric states may provide an example of bound entangled states that have negative partial transposes \cite{DiVincenzo2000}. We can restrict our attention only to states that are symmetric in some manner, for example the well-known Werner and isotropic states, and exploit that symmetry to compute the convex roof of certain entanglement monotones on those families of states. For example, the entanglement of formation has been computed for Werner states \cite{Vollbrecht2001} and isotropic states \cite{Terhal2000}. Convex roofs of some generalizations of the concurrence \cite{Gour2005} have been computed for isotropic states as well \cite{Eltschka2015,Sentis2016}.

In this work, we expand on existing methods \cite{Vollbrecht2001,Terhal2000} to compute the convex roofs of many more entanglement monotones for these classes of symmetric states and more. In particular, we show that our method can be used to compute the convex roof on Werner states for \emph{all} possible entanglement monotones on pure states.  We also compute the convex roof of the Vidal monotones in Eq.~\eqref{eq:vidalmonotones} and certain other entanglement monotones for isotropic states. We also extend these methods to compute the convex roof on larger classes of symmetric states.

While entanglement monotones are important for quantifying entanglement within states, it is also important to characterize which transformations between states can be performed via LOCC \emph{deterministically}. For bipartite pure states, this is completely characterized by majorization of the vectors of Schmidt coefficients \cite{Nielsen1999}, or equivalently by the Vidal monotones \cite{Vidal2000b}. Only a finite number of entanglement measures are needed to determine the convertibility of bipartite pure states, but an infinite number of entanglement measures are needed to completely determine convertibility of mixed states \cite{Gour2005a}. To characterize the convertibility of mixed states, we can instead make use of entanglement \emph{conversion witnesses} \cite{Girard2015a,Gour2013}. An entanglement conversion witness is a function of two bipartite quantum states whose value ``detects'' when one state can be converted into another. For example, a \emph{no-go} entanglement conversion witness is a function $W(\rho,\sigma)$ such that $W(\rho,\sigma)<0$ implies that $\rho$ cannot be converted to $\sigma$ with a deterministic LOCC operation. Similarly, a \emph{go} entanglement conversion witness is a function $W(\rho,\sigma)$ such that $W(\rho,\sigma)\geq0$ implies the existence of a deterministic LOCC protocol that converts $\rho$ into $\sigma$. A witness is \emph{complete} if it is both a go and a no-go witness.

In Ref.~\cite{Gheorghiu2008}, it was shown that a bipartite pure state $\ket{\psi}$ can be converted into a bipartite mixed state $\rho$ if and only if
\[
 E_k(\psi) \geq\sum_{i}p_iE_k(\varphi_i)
\]
holds for all $k$ and all decompositions $\rho=\sum_i p_i \ketbra{\varphi_i}{\varphi_i}$. This necessary and sufficient condition for LOCC transformation can be encoded into the following complete witness:
\[
 W(\psi,\rho) = \max_{\{p_i,\varphi_i\}}\min_k \biggl( E_k(\psi) -\sum_{i}p_iE_k(\varphi_i)\biggr).
\]
It holds that $W(\psi,\rho)\geq0$ if and only if $\ket{\psi}$ can be converted into $\rho$ via LOCC. Although this function cannot be computed for arbitrary mixed states, we can make extensive use of symmetry to compute it in the case when $\rho$ is highly symmetric (e.g.~Werner or isotropic). In the final section of this paper, we show how to compute a class of entanglement transformation witnesses for pure to mixed bipartite state conversion in the case when the target mixed state is symmetric. 

The remainder of this paper is structured as follows. The necessary background for constructing convex roof functions, the definition of the Werner and isotropic states, and other preliminary matter are presented in Sec.~\ref{sec:prelim}. Convex roofs of certain entanglement monotones are evaluated on Werner and isotropic states in Sec.~\ref{sec:convexroofs}. An entanglement transformation witness for pure to mixed state conversion is presented in Sec.~\ref{sec:witnesses}, where it is also shown how to evaluate this witness when the target state is a Werner state or an isotropic state. %Concluding remarks are given in section~\ref{sec:conclude}.

%%%%%%%%%%%%%%%%%%%%%%%%%%%%%%%%%%%%%%%%%%%%%%%%%%%%%%%%%%%%%%%%

\section{Preliminaries}
\label{sec:prelim}

In this section, we review the notion of a convex roof of an arbitrary function. The  details from Ref.~\cite{Vollbrecht2001} that are necessary for computing the convex roofs of functions under generalized symmetry are summarized. We also review the types of bipartite symmetries that we will analyze, in particular the Werner and isotropic states and generalizations of these symmetries.

%~~~~~~~~~~~~~~~~~~~~~~~~~~~~~~~~~~~~~~~~~~~~~~~~~~~~~~~
\subsection{Convex roofs and symmetry}
In the following, we use the notation $\overline{\RR}=\RR\cup\{+\infty\}$. Let $K$ be a compact set, $M\subset K$, and let $f\:M\rightarrow\RR$. The \emph{convex roof} of $f$ over $K$ is the function $\cohat{f}:K\rightarrow\overline{\RR}$ defined as
\begin{equation}\label{eq:cohatf}
 \cohat{f}(x) = \inf\biggl\{\sum_{i}p_if(y_i)\,\bigg|\,y_i\in M,\, \sum_{i}p_i y_i =x \biggr\},
\end{equation}
for any $x\in\co(M)$ in the convex hull of $M$. The infimum in Eq.~\eqref{eq:cohatf} is taken over all convex combinations with $p_i\geq 0$ and $\sum_{i}p_{i}=1$. Note that $\widehat{f}(x)=\infty$ if $x\notin\co(M)$.

Let $\calG$ be a compact group with a $\calG$-action $g\cdot x$ on~$K$ that preserves convex combinations (i.e.\ $g\cdot(tx+(1-t)y)=tg\cdot x + (1-t)g\cdot y$ for any $x,y\in K$ and any $t\in[0,1]$). Then the $\calG$-\emph{twirling} operator $\calT_\calG:K\rightarrow K$ is defined as
\begin{equation}
 \calT_\calG (x) = \int_\calG dg\, g\cdot x,
\end{equation}
for all $x\in K$, where the integral is taken over the Haar measure of the group. If $\calT_\calG(y)=x$ then we say that~$y$ \emph{twirls} to $x$ under $\calG$. The $\calG$-invariant elements $x\in K$ are exactly those that satisfy $\calT_\calG(x)=x$, and the subset of $\calG$-invariant elements of $K$ is denoted as $\calT_\calG(K)$.

Given any function $f\:M\rightarrow \RR$ on a subset $M\subset K$, we define the function $f_\calG\:\calT_\calG(K)\rightarrow\overline{\RR}$ as
\begin{equation}\label{eq:fG}
 f_\calG(x) = \inf \left\{f(y)\,\middle|\,y\in M,\, \calT_\calG(y)=x \right\}
\end{equation}
for all $x\in\calT_\calG(K)$. As the following theorem shows, this definition allows us to find a different expression for the convex roof of a function $f\:M\rightarrow\overline{\RR}$ evaluated on $\calG$-invariant elements of $K$. This is the primary tool that we will use to compute convex roof entanglement monotones on the Werner and isotropic states.

\begin{theorem}[Sec.\ IV.A in Ref.~\cite{Vollbrecht2001}]\label{thm:coco}
 Let $\calG$ be a compact group and $K$ be a compact convex set with a $\calG$-action that preserves convex combinations, and let $f\:M\rightarrow\overline{\RR}$ be a function on a subset $M\subset K$. It holds that
 \begin{equation}
  \cohat{f}(x) = \cohat{f_\calG}(x)
 \end{equation}
for all $x\in\calT_\calG(K)$.
\end{theorem}

To compute the convex roof $\cohat{f}$ of a function $f$ on the $\calG$-invariant elements of $K$, the result of Theorem \ref{thm:coco} implies that we can simplify the computation by first minimizing $f$ over all $y\in M$ that twirl to $x$. Computing the convex roof of the resulting function yields the desired result. This computation is simplified greatly if $f_\calG$ is already convex as a function of $\calG$-invariant elements, in which case $\cohat{f}(x)$ reduces to $f_\calG(x)$. Note that both $f_\calG$ and $\widehat{f_\calG}$ are functions on the convex subset $\calT_\calG(K)\subset K$ of elements that are invariant under the action of $\calG$.

One basic feature of convex roof functions is the existence of `linear sections' in the roof function whenever the infimum in Eq.~\eqref{eq:cohatf} is found at a non-trivial convex combination. The result of Lemma \ref{lem:nonsymmelems} (which is proven in Ref.~\cite{Vollbrecht2001}) allows us to compute convex roof functions on some elements that are not necessarily symmetric with respect to the group action.
\begin{lemma}\label{lem:nonsymmelems}
Suppose $x=\sum_ip_ix_i\in K$ is a convex combination of elements $x_i\in M$ with $p_i>0$ for each $i$ such that $\widehat{f}(x)=\sum_ip_if(x_i)$ is minimized. Then $\widehat{f}$ is linear on the convex hull of $\{x_i\}$. That is, it holds that
 \begin{equation}
  \widehat{f}\Bigl(\sum_i t_i x_i\Bigr) = \sum_{i} t_i f(x_i)
 \end{equation}
for all $t_i\in[0,1]$ satisfying $\sum_i t_i = 1$.
\end{lemma}

In our analysis here, we compute the convex roof of entanglement monotones on pure states for Werner and isotropic states. The minimizing sets will usually be an entire orbit of some pure state under the local-unitary group action. Every pure state in these orbits has the same amount of entanglement under any entanglement monotone, since they differ only by a local unitary. Hence the convex roof of any entanglement monotone will be constant on the convex hull of these orbits. This gives a fairly large class of non-symmetric states for which we can compute the exact value of many different entanglement monotones.

%~~~~~~~~~~~~~~~~~~~~~~~~~~~~~~~~~~~~~~~~~~~~~~~~~~~~~~~
\subsection{Bipartite entanglement symmetry}
\label{sec:bipartitesymmetries}

In this following section, we recall some well-known examples of groups that are used in the study of bipartite quantum entanglement. Let $d\geq 2$ be an integer and consider bipartite states on $\CC^d\otimes\CC^d$. The convex set of interest here is the set of normalized density operators $\D(\CC^d\otimes\CC^d) = \{\rho\,|\, \rho\geq 0,\,\Tr\rho=1\}$. We are interested in computing the convex roof of entanglement monotones that are defined on the pure states
\[
 \bigl\{\ketbra{\psi}{\psi}\,\big|\, \ket{\psi}\in\CC^d\otimes\CC^d,\, \norm{\ket{\psi}}=1\bigr\}\subset \D(\CC^d\otimes\CC^d).
\]
It is well known that any entanglement monotone on pure states must be a symmetric, concave function of the Schmidt coefficients of the pure states. The primary examples of symmetric states that we study in this paper are the well-known Werner states \cite{Werner1989} and isotropic states \cite{Horodecki1999}. 

For the remainder of this paper we assume that $d\geq 2$ and we only consider bipartite states on $\CC^d\otimes\CC^d$. We consider classes of states that are symmetric with respect to different subgroups of the group of local unitaries
\[
 \LU = \{U\otimes V \, |\, U,V\in\U(d)\}.
\]
Given a subgroup $\calG\subset\LU$, determining which states are invariant under $\calG$ amounts to computing the commutant of $\calG$,
\[
 \calG' =\{A\in\mathcal{B}(\CC^d\otimes\CC^d)\,|\, [A,g]=0 \text{ for all }g\in\calG\},
\]
where $\mathcal{B}(\CC^d\otimes\CC^d)$ denotes the set of linear operators on the tensor product space (i.e.\ the set of $d^2\times d^2$ matrices). The commutant $\calG'$ is the subspace of operators that commute with every element of $\calG$. The twirling operator $\calT_\calG$ can be viewed as the projection operator onto the commutant of $\calG$. To determine $\calG'\cap \D(\CC^d\otimes\CC^d)$, i.e.\ the family of states that are invariant under this action, it is useful to find an orthogonal basis of operators for $\calG'$ and express the states as combinations of those basis elements. Finally, note that for any $\calG\subseteq\LU$ the twirling operation $\calT_\calG$ is an LOCC operation, since it consists of a convex mixture of local unitary channels. 

%~~~~~~~~~~~~~~~~~~~~~~~~~~~~
\subsubsection{Werner states}
The $d\times d$ \emph{Werner states} \cite{Werner1989} are those that commute with all unitaries of the form $U\otimes U$ for some~$U\in\U(d)$. That is, Werner states are those which are invariant under the subgroup $\{U\otimes U\,|\, U\in\U(d)\}$. The corresponding twirling operator is
 \[
  \calT_{\wer} (\rho) = \int_{\U(d)} dU\, U\otimes U \rho (U\otimes U)^\dagger ,
 \]
where the integral is taken over the Haar measure of the group $\U(d)$ of $d\times d$ unitary matrices. The commutant of this group is spanned by $\{\mathds{1},W\}$, where $\mathds{1}$ is the identity operator and $W$ is the swap operator defined by
$W = \sum_{i,j=1}^d\ketbra{ij}{ji}$. The swap operator is both unitary and Hermitian, having eigenvalues $1$ and $-1$ and satisfying $W^2=\mathds{1}$. Let $W_{\!\!+}$ and $W_{\!\!-}$ denote the projectors onto the subspaces spanned by the positive and negative eigenvectors of $W$, respectively, such that $W=W_{\!\!+}-W_{\!\!-}$. The Werner states can then be parametrized by
\begin{equation}
 \rho_{\wer}(\walpha) = \walpha \tfrac{1}{\binom{d}{2}} W_{\!\!-} + (1-\walpha)\tfrac{1}{\binom{d+1}{2}}W_{\!\!+}
\end{equation}
for $\walpha\in[0,1]$. These states are entangled for $\walpha\in[\tfrac{1}{2}, 1]$ and separable otherwise \cite{Vollbrecht2001,Watrous2016}. Furthermore, it holds that $\mathcal{T}_{\wer}(\sigma) = \rho_{\wer}(\Tr[\sigma W_{\!\!-}])$ for all states $\sigma$.

%~~~~~~~~~~~~~~~~~~~~~~~~~~~~
\subsubsection{Isotropic states}
The $d\times d$ \emph{isotropic states} \cite{Horodecki1999} are those invariant under the subgroup $\{U\otimes \overline{U}\,|\, U\in\U(d)\}$. The corresponding twirling operator is
\[
 \calT_{\isot} (\rho) = \int_{\U(d)} dU\, U\otimes \overline{U} \rho (U \otimes\overline{U})^\dagger .
\]
The commutant of this group is spanned by $\{\mathds{1},\Phi_d\}$, where $\Phi_d = \frac{1}{d}\sum_{i,j=1}^d \ketbra{ii}{jj}$ is the projection operator onto the maximally entangled pure state $\frac{1}{\sqrt{d}}\sum_{i=1}^d\ket{ii}$ of two qudits. This commutant is exactly the partial transpose of the space from the Werner states \cite{Audenaert2002}. The isotropic states can be parametrized by
\begin{equation}\label{eq:isotropicstatesdef}
 \rho_{\isot}(\wbeta) = \wbeta \Phi_d + (1-\wbeta) \frac{\mathds{1}-\Phi_d}{d^2-1}
\end{equation}
for $\wbeta\in[0,1]$. The isotropic states are entangled for $\wbeta\in[\tfrac{1}{d}, 1]$ and separable otherwise \cite{Vollbrecht2001,Watrous2016}. Furthermore, it holds that $\mathcal{T}_{\isot}(\sigma) = \rho_{\isot}(\Tr[\sigma \Phi_d])$ for all states $\sigma$.

%~~~~~~~~~~~~~~~~~~~~~~~~~~~~
\subsubsection{\texorpdfstring{$OO$}{OO}-invariant states}

One way to generalize the isotropic and Werner states to larger classes of symmetric states is to consider the $OO$-invariant states \cite{Vollbrecht2001}. These are the states that are invariant under $\{U\otimes U\,|\,U\in\Orth(d)\}$, where $\Orth(d)\subset\U(d)$ is the group of orthogonal operators. Since the orthogonal matrices are the unitaries that satisfy $\overline{U}=U$, this group is a subgroup of both the isotropic group and the Werner group of local unitaries. The corresponding $OO$-twirling operator is defined as
\[
 \calT_{\Orth} (\rho) =\int_{\Orth(d)} dU\, U\otimes U \rho (U\otimes U)^\dagger .
\]
The commutant of this group is spanned by $\{\mathds{1},W,\Phi_d\}$ \cite[section II D]{Vollbrecht2001}. The $OO$-invariant states can be parametrized as
\begin{equation}\label{eq:OOinvstates}
 \rho_{\Orth}(\walpha,\wbeta)  = \walpha\tfrac{1}{\binom{d}{2}} W_{\!\!-} + \wbeta\Phi_d + (1-\walpha-\wbeta) \tfrac{1}{\binom{d+1}{2}-1}(\mathds{1}-\Phi_d-W_{\!\!-})
\end{equation}
for $\walpha,\wbeta\in[0,1]$ satisfying $\walpha+\wbeta\leq 1$. The $OO$-invariant states that are separable (and also positive under partial transposition) \cite{Vollbrecht2001} are those in the rectangle $(\walpha,\wbeta)\in[0,\frac{1}{2}]\times[0,\frac{1}{d}]$. The Werner states are $OO$-invariant states for which $\wbeta=\frac{2(1-\walpha)}{d(d+1)}$ and the isotropic states are those for which $\wbeta=1-\frac{2(d+1)}{d}\walpha$. A schematic of the $OO$-invariant states is shown in Fig~\ref{fig:ooinv1}. 

The entanglement of formation and the asymptotic relative entropy of entanglement of $OO$-invariant states have been computed \cite{Vollbrecht2001,Audenaert2002}. In Sec.~\ref{sec:convexroofs}, we show how to compute almost any convex roof monotone on the $OO$-invariant states.

\begin{figure}\centering
\includegraphics{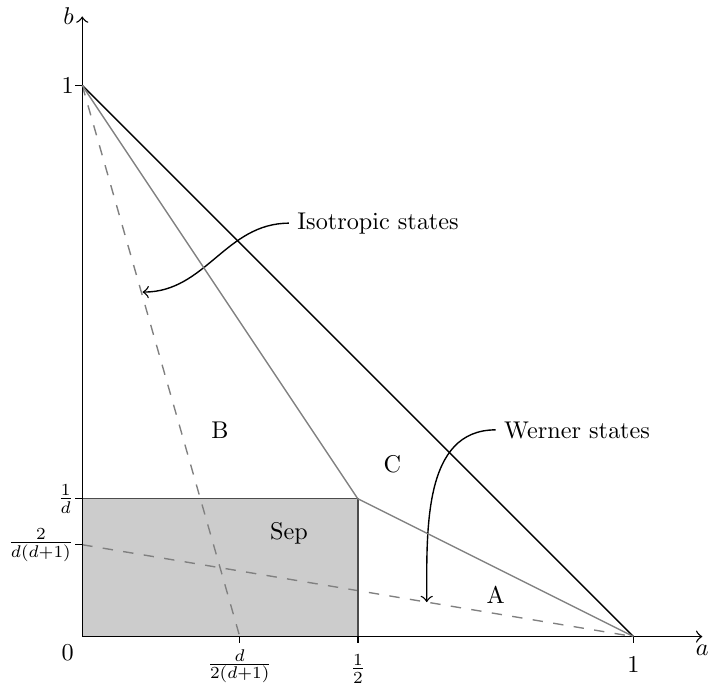}
\caption{Schematic of the $OO$-invariant states $\rho_{\Orth}(\walpha,\wbeta)$, as defined in Eq.~\eqref{eq:OOinvstates}. The shaded region represents the separable states. The one-dimensional subfamilies of Werner and isotropic states are also shown. Convex roof entanglement monotones can be computed for states in regions A and B, as discussed in Sec.~\ref{sec:othersyms}. It remains unknown how to compute convex roofs on states in region C for an arbitrary entanglement monotone.}
\label{fig:ooinv1}
\end{figure}

%~~~~~~~~~~~~~~~~~~~~~~~~~~~~
\subsubsection{Phase-permutation-invariant states}
\label{sec:phaseperminv}

Other subgroups of $\U(d)$ lead to further generalizations of the Werner and isotropic states. One possible subgroup that leads to two-parameter families of symmetric states is the following. Consider the subgroup of `phase-permutation' unitary matrices defined by
\begin{equation}\label{eq:Gunitarygroup}
 G=\{P_\pi U \,|\, \pi\in\mathcal{S}_d,\, U\in\U(d)\text{ is diagonal}\},
\end{equation}
where $\mathcal{S}_d$ is the symmetric group and $ P_\pi =\sum_{i=1}^d \ket{\pi(i)}\bra{i}$ is the permutation matrix for $\pi\in\mathcal{S}_d$. If we denote the group of diagonal unitary matrices by $N\simeq\U(1)^{\times d}$, we see that $N$ is a normal subgroup of~$G$. The group $G$ of phase-permutation unitaries can be viewed as the semi-direct product $G=N\rtimes\mathcal{P}_d$, where $\mathcal{P}_d=\{P_\pi\,|\, \pi\in\mathcal{S}_d\}$ denotes the group of $d\times d$ permutation matrices. This is also the subgroup of unitaries that have exactly one nonzero entry in each row and column. 

%~~~~~~~~~~~~~~~~~~~~~~~~~~~~
\vspace{10pt}
\paragraph{Phase-permutation Werner states}Consider the family of Werner-type states that are invariant under $\{U\otimes U\,|\,U\in G\}$, where $G$ is the group of phase-permutation matrices defined in Eq.~\eqref{eq:Gunitarygroup}. Such states will be referred to in this paper as \emph{phase-permutation Werner states}. This class of states was first introduced in Ref.~\cite{DiVincenzo2000} and used in Ref.~\cite{Girard2015a}. The corresponding twirling operation is
\[
 \calT^G_{\wer} (\rho) =\int_{G} dU\, U\otimes U \rho (U\otimes U)^\dagger .
\]
The commutant of this group is spanned by $\{\mathds{1},W,Q\}$ \cite[Sec.~II]{DiVincenzo2000}, where $Q$ is the projection operator
\begin{equation}\label{eq:Qdef}
 Q=\sum_{i=1}^d \ket{ii}\bra{ii}
\end{equation}
that satisfies $[Q,W_{\!\!\pm}]=0$, $QW_{\!\!-}=0$, and $QW_{\!\!+}=Q$. This family of states can be parametrized by
\begin{equation}
 \rho^G_{\wer}(\walpha,\wbeta)  = \walpha\frac{1}{\binom{d}{2}} W_{\!\!-} + \wbeta\frac{1}{\binom{d}{2}}(W_{\!\!+}-Q) + (1-\walpha-\wbeta) \frac{1}{d}Q
\end{equation}
for $\walpha,\wbeta\in[0,1]$ satisfying $\walpha+\wbeta\leq 1$. For all states~$\rho$, it holds that $\calT^G_{\wer} (\rho)=\rho^G_{\wer}(\walpha,\wbeta)$, where $\walpha=\Tr[\rho W_{\!\!-}]$ and $\wbeta=\Tr[\rho (W_{\!\!+}-Q)]$. The Werner states form a subfamily of this class. A schematic of the phase-permutation Werner states is depicted in Fig.~\ref{fig:pphaspermwer1}.

\begin{figure}[]
\centering
\includegraphics{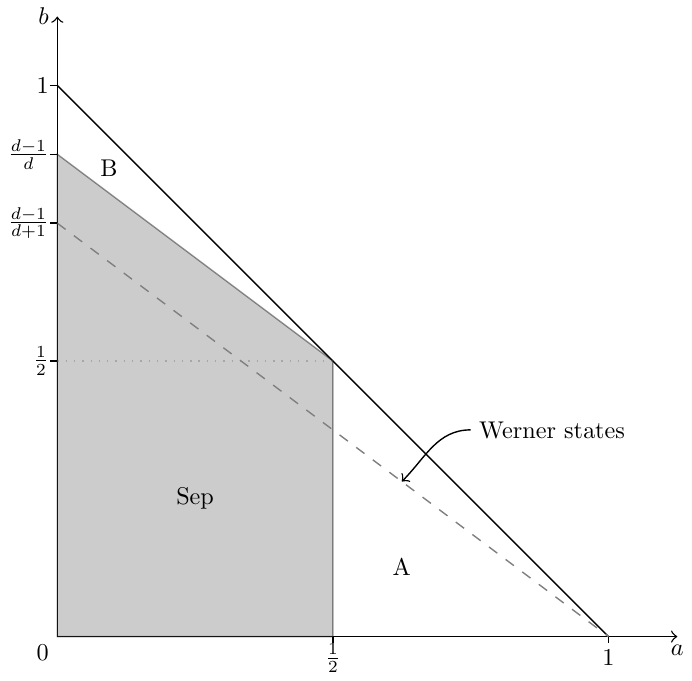}
\caption{Schematic of the phase-permutation Werner states. The separable region is shown in gray. The one-dimensional family of states with $\wbeta=\frac{d-1}{d+1}(1-\walpha)$ is made up of the well-known Werner states. As shown in Sec.~\ref{sec:othersyms}, the convex roof of any entanglement monotone can be computed for any state in region A. It remains unknown how to compute convex roofs on states in region B for an arbitrary entanglement monotone.}
\label{fig:pphaspermwer1}
\end{figure}

%~~~~~~~~~~~~~~~~~~~~~~~~~~~~
\vspace{10pt}
\paragraph{Phase-permutation isotropic states}
Similarly, we can consider the family of isotropic-type states that are invariant under $\{U\otimes \overline{U}\,|\,U\in G\}$. We refer to these as the \emph{phase-permutation isotropic states}. These states have been studied by others \cite{Eltschka2013,Eltschka2015,Sentis2016} who have called them the \emph{axisymmetric states}. The corresponding twirling operation is
\[
 \calT^G_{\isot} (\rho) =\int_{G} dU\, U\otimes \overline{U} \rho (U \otimes \overline{U})^\dagger .
\]
The commutant of this group is spanned by $\{\mathds{1},\Phi_d,Q\}$. The elements of this commutant are exactly obtained from the partial transposes of the elements of the commutant of the phase-permutation Werner group presented in the previous paragraph. The family of phase-permutation isotropic states can be parametrized as
\begin{equation}\label{eq:phasepermisostates}
 \rho^G_{\isot}(\walpha,\wbeta)  = \wbeta\Phi_d + \walpha\tfrac{1}{d-1}(Q-\Phi_d) + (1-\walpha-\wbeta) \tfrac{1}{d(d-1)}(\mathds{1}-Q)
\end{equation}
for $\walpha,\wbeta\in[0,1]$ satisfying $\walpha+\wbeta\leq 1$. For all states~$\rho$, it holds that $\calT^G_{\isot} (\rho)=\rho^G_{\isot}(\walpha,\wbeta)$ where $\wbeta=\Tr[\rho\Phi_d]$ and $\walpha=\Tr[\rho (Q-\Phi_d)]$. The isotropic states form a subfamily of this class. A schematic of the phase-permutation isotropic states is depicted in Fig.~\ref{fig:pphaspermiso1}. 

\begin{figure}\centering
\includegraphics{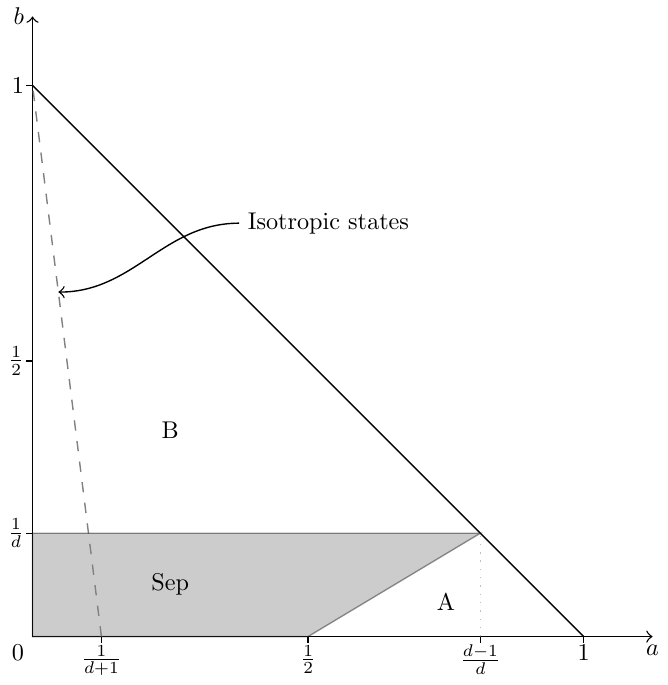}
\caption{Schematic of the phase-permutation isotropic states. The separable region is shown in gray. The one-dimensional family of states with $\wbeta=1-(d+1)\walpha$ is made up of the isotropic states. As shown in Sec.~\ref{sec:othersymms}, the convex roof of any entanglement monotone can be computed for any state in region~B. It remains unknown how to compute convex roofs on states in region A for an arbitrary entanglement monotone.}
\label{fig:pphaspermiso1}
\end{figure}

%%%%%%%%%%%%%%%%%%%%%%%%%%%%%%%%%%%%%%%%%%%%%%%%%%%%%%%%%%%%%%%%

\section{Convex roof entanglement monotones for symmetric states}
\label{sec:convexroofs}
In this section, we compute the convex roofs of entanglement monotones evaluated on Werner and isotropic states. For Werner states, we compute this for any monotone. For isotropic states, we compute the convex roofs of the Vidal monotones and generalize the computation to certain classes of other monotones. 

%~~~~~~~~~~~~~~~~~~~~~~~~~~~~~~~~~~~~~~~~~~~~~~~~~~~~~~~
\subsection{Werner states}
In this subsection we present a general method for computing convex roofs of entanglement monotones evaluated on the Werner states of a $d\times d$ bipartite system. For any $\walpha\in[0,1]$, consider the minimum entanglement of all pure states that twirl to $\rho_{\wer}(\walpha)$ under this action as in Eq.~\eqref{eq:fG}. Given an arbitrary entanglement monotone $E$ on pure states, define the function $E_{\wer}\:[0,1]\rightarrow \RR$ as
\begin{equation}\label{eq:Ewerpsi}
 E_{\wer}(\walpha) =\min \{E(\psi)\,|\, \braketinner{\psi}{W_{\!\!-}}{\psi}=\walpha\}.
\end{equation}
If we can evaluate Eq.~\eqref{eq:Ewerpsi} for a given entanglement monotone $E$, then we may make use of Theorem \ref{thm:coco} to compute the convex roof of $E$ on Werner states by computing $\cohat{E_{\wer}}$. This result is greatly simplified if $E_{\wer}$ is already convex as a function of $\walpha$. 

\begin{theorem}\label{lem:Ewerequal}
 Let $E$ be an entanglement monotone on pure states. For all $\walpha\in[0,1]$, it holds that
 \begin{equation}\label{eq:Ewerpsialpha}
  E_{\wer}(\walpha) = E(\psi_\walpha),
 \end{equation}
where $E_{\wer}$ is the function as defined in Eq.~\eqref{eq:Ewerpsi}, and $\ket{\psi_\walpha}$ are the pure states defined by
\begin{equation}\label{eq:psigam1}
 \ket{\psi_\walpha} = \left(\sqrt{1-2\walpha}\,\ket{1} + \sqrt{2\walpha}\,\ket{2}\right)\otimes \ket{2}
\end{equation}
whenever $\walpha\in[0,\tfrac{1}{2}]$ and 
\begin{equation}\label{eq:psigam2}
 \ket{\psi_\walpha} = \sqrt{\tfrac{1}{2}+\sqrt{\walpha(1-\walpha)}}\,\ket{12} - \sqrt{\tfrac{1}{2}-\sqrt{\walpha(1-\walpha)}}\,\ket{21}
\end{equation}
whenever $\walpha\in[\tfrac{1}{2},1]$.
\end{theorem}

Note that the pure states $\ket{\psi_\walpha}$ twirl to the Werner state $\rho_{\wer}(\walpha)$. Indeed, a straightforward calculation shows that $\braketinner{\psi_{\walpha}}{W_{\!\!-}}{\psi_{\walpha}} = \walpha$ for all $\walpha$. In particular, Theorem \ref{lem:Ewerequal} states that the pure states $\ket{\psi_\walpha}$ are in fact optimal in the computation in Eq.~\eqref{eq:Ewerpsi} for \emph{every} possible entanglement monotone. This is a generalization of the statement in Ref.~\cite[Sec.\ IV.C]{Vollbrecht2001}, where the convex roof of the entanglement of formation was computed for Werner states. The proof of Theorem~\ref{lem:Ewerequal}, which can be found in Appendix \ref{sec:werentproof}, is quite technical and follows the method used in Ref.~\cite{Vollbrecht2001}.

From Theorems \ref{thm:coco} and \ref{lem:Ewerequal}, it follows that $\cohat{E}(\rho_{\wer}(\walpha)) = \cohat{E_{\wer}}(\walpha)$. The family of Werner states is convex and 
\[
 t\rho_{\wer}(a_1) + (1-t)\rho_{\wer}(a_2) = \rho_{\wer}\bigl(ta_1+(1-t)a_2\bigr).
\]
Hence the computation of $\cohat{E_{\wer}}(\walpha)$ is greatly simplified if $E_{\wer}$ is already convex as a function of $\walpha$ (as it is for the entanglement of formation). Otherwise there are simple procedures for computing the convex roof of a function of a single variable. Even if the convex roof of $E_{\wer}$ as a function of $\walpha$ cannot be computed for a particular entanglement monotone $E$, the formula in Eq.~\eqref{eq:Ewerpsialpha} still gives an upper bound for $\cohat{E}$ on Werner states since $\cohat{E}(\rho_{\wer}(\walpha))=\cohat{E_{\wer}}(\walpha)\leq E_{\wer}(\walpha)$ always holds. 

%~~~~~~~~~~~~~~~~~~~~~~~~~~~~
\subsubsection{Entanglement of formation}
The entanglement of formation~\cite{Bennett1996b} is one well-known convex roof entanglement monotone. This is defined as $E_F(\rho) = \cohat{E}(\rho)$ for mixed states $\rho$, where $E$ is the entropy of entanglement on pure states $E(\psi)=H(\vec{\lambda})$, $H$ is the Shannon entropy, and $\vec{\lambda}$ is the vector of Schmidt coefficients of $\ket{\psi}$. When $\walpha\in[\frac{1}{2},1]$, the entropy of entanglement of $\ket{\psi_\walpha}$ is given by 
\begin{equation}\label{eq:entformpsialpha}
 E(\psi_\walpha) = h\bigl(\tfrac{1}{2}-\sqrt{\walpha(1-\walpha)}\bigr),
\end{equation}
where $h(t)=-t\log t- (1-t)\log(1-t)$ is the binary entropy function. Note that the function in Eq.~\eqref{eq:entformpsialpha} is convex as a function of $\walpha$, so it follows that 
 \begin{equation}
  E_F(\rho_{\wer}(\walpha)) = \left\{\begin{array}{ll}
                                     0, & \walpha\in[0,\frac{1}{2}]\\
                                     h\bigl(\tfrac{1}{2}-\sqrt{\walpha(1-\walpha)}\bigr), & \walpha\in[\frac{1}{2},1].
                                    \end{array}
\right.
 \end{equation}
This matches the result found in Ref.~\cite{Vollbrecht2001}.

%~~~~~~~~~~~~~~~~~~~~~~~~~~~~
\subsubsection{Vidal monotones}

Consider now the Vidal monotones $E_k$ on pure states. Evaluating the convex roof of these monotones on the Werner states can be done easily, because $E_{k,\wer}(\walpha)$ is already convex as a function of $\walpha$. 
\begin{theorem}
 Consider the convex roof of the Vidal monotones $E_k$ on Werner states. The first Vidal monotone reduces to
 \begin{equation}
  \cohat{E_1}(\rho_{\wer}(\walpha)) = \left\{\begin{array}{ll}
                                     0, & \walpha\in[0,\frac{1}{2}]\\
                                     \frac{1}{2}-\sqrt{\walpha(1-\walpha)}, & \walpha\in[\frac{1}{2},1].
                                    \end{array}
\right.
 \end{equation}
 For $k>1$, $\cohat{E_k}(\rho_{\wer}(\walpha)) =0$ for all $\walpha$.
\end{theorem}
 In particular, the convex roof of the $k\textsuperscript{th}$ Vidal monotone vanishes for all Werner states when $k\neq 1$. Indeed, it holds that $E_k(\psi_\walpha)=0$ for all $\walpha$ if $k>1$, since the Schmidt vector of $\ket{\psi_\walpha}$ has at most two nonzero components. For $\walpha\in[\frac{1}{2},1]$, note that
\[
 E_1(\psi_\walpha) = \frac{1}{2}-\sqrt{\walpha(1-\walpha)},
\]
which is already convex as a function of $\walpha$.

%~~~~~~~~~~~~~~~~~~~~~~~~~~~~
\subsubsection{Renyi entropies}
The result of Theorem \ref{lem:Ewerequal} can also be used to compute the convex roofs of Renyi entropies \cite{Hastings2010} of entanglement evaluated on Werner states. For $\alpha>0$ with $\alpha\neq 1$, the Renyi-$\alpha$ entropy of entanglement is defined as $E_\alpha(\vec{\lambda})=\frac{1}{1-\alpha}\log\bigl(\sum_{i}\lambda_i^\alpha\bigr)$ for pure states with Schmidt vector $\vec{\lambda}$. These are in fact valid entanglement monotones on pure states when $\alpha\in[0,1]$ \cite{Vidal2000a}. The form of Eq.~\eqref{eq:Ewerpsi} for these monotones reduces to $E_{\alpha,\wer}(\walpha)=0$ when $a\in[0,\frac{1}{2}]$ and
\begin{multline}\label{eq:renyiwerner}
 E_{\alpha,\wer}(\walpha) = \\ \tfrac{1}{1-\alpha}\log\left(\left(\tfrac{1}{2}+\sqrt{\walpha(1-\walpha)}\right)^\alpha+\left(\tfrac{1}{2}-\sqrt{\walpha(1-\walpha)}\right)^\alpha\right)
\end{multline}
 when $\walpha\in[\frac{1}{2},1]$. Numerical evidence suggests that Eq.~$\eqref{eq:renyiwerner}$ is strictly convex whenever $\alpha>1$, and that Eq.~$\eqref{eq:renyiwerner}$ is strictly concave on the interval $\walpha\in[\frac{1}{2},1]$ whenever $\alpha<\frac{1}{2}$. Thus $\cohat{E_{\alpha}}(\rho_{\wer}(\walpha)) =E_{\alpha,\wer}(\walpha)$ for $\alpha>1$ and $\cohat{E_{\alpha}}(\rho_{\wer}(\walpha)) =\max\{0,(2\walpha-1)\log 2\}$ for $\alpha<\frac{1}{2}$.

%~~~~~~~~~~~~~~~~~~~~~~~~~~~~~~~~~~~~~~~~~~~~~~~~~~~~~~~
\subsection{Isotropic states}
In this section we present a general method for computing convex roofs of entanglement monotones evaluated on the isotropic states of a $d\times d$ bipartite system. In particular, we show explicit formulas for the convex roofs of the Vidal monotones, as we did for the Werner states in the previous section. Using majorization, the result for the Vidal monotones is used to find a simple lower bound for any entanglement monotone on isotropic states. An outline for computing the convex roof of the Renyi entropies on isotropic states is also presented. Detailed proofs can be found in Appendix \ref{sec:isoentproof}.

The isotropic states $\rho_{\isot}(\wbeta)$ defined in Eq.~\eqref{eq:isotropicstatesdef} are the states invariant under the action $U\cdot\rho=U\otimes \overline{U}\rho (U\otimes \overline{U})^\dagger $ from the $d$-dimensional unitary matrices $U$. Similar to our analysis of Werner states, for any $\wbeta\in[0,1]$ we consider the minimum entanglement of all pure states that twirl to $\rho_{\isot}(\wbeta)$ under this action as follows. Given an arbitrary entanglement monotone $E$ on pure states, define the function $E_{\isot}\:[0,1]\rightarrow \RR$ by
\begin{equation}\label{eq:Eisotpsi}
 E_{\isot}(\wbeta) =\min \bigl\{E(\psi)\,\big|\, \braketinner{\psi}{\Phi_d}{\psi}=\wbeta\bigr\}.
\end{equation}
If we can determine a closed-form expression of Eq.~\eqref{eq:Eisotpsi} for a given entanglement monotone~$E$, we can make use of Theorem \ref{thm:coco} to compute the convex roof of~$E$ on isotropic states by computing $\cohat{E_{\isot}}$. This result is greatly simplified if $E_{\isot}$ is already convex as a function of $\wbeta$. We use the result of the following lemma to simplify computations.

\begin{lemma}\label{lem:Eisoequal}
 Let $E$ be an entanglement monotone on pure states. For all $\wbeta\in[\frac{1}{d},1]$, it holds that
 \begin{equation}\label{eq:Eisoequal}
  E_{\isot}(\wbeta) = \min\biggl\{E(\vec\lambda)\,\bigg|\, \sum_{i=1}^d\sqrt{\lambda_i} =\sqrt{d\wbeta}\biggr\},
 \end{equation}
where the infimum is taken over all Schmidt vectors satisfying the condition. Furthermore, $E_{\isot}(\wbeta)=0$ whenever $\wbeta\in[0,\frac{1}{d}]$.
\end{lemma}

A closed-form expression for $E_{\isot}$ in the right-hand side of Eq.~\eqref{eq:Eisoequal} can actually be computed for specific monotones~$E$, which we show in the remainder of this section. In particular, we compute $E_{\isot}$ in the cases when $E$ is a Vidal monotone or an entropy-type monotone. The proof of Lemma~\ref{lem:Eisoequal}, which is a generalization of the result in Ref.~\cite{Terhal2000}, is quite technical and can be found in Appendix~\ref{sec:isoentproof}. 

%~~~~~~~~~~~~~~~~~~~~~~~~~~~~
\subsubsection{Vidal monotones}
Here we present the results for evaluating the convex roofs of the Vidal monotones \eqref{eq:vidalmonotones} on isotropic states. The Schmidt vector that minimizes $E_{k,\isot}$ in Eq.~\eqref{eq:Eisoequal} will be of the form.
\begin{equation}\label{eq:lambdaisotoptmain}
 \vec{\lambda} = \bigl(\underbrace{t,\dots,t}_{k},\underbrace{\tfrac{1-kt}{d-k},\dots,\tfrac{1-kt}{d-k}}_{d-k}\bigr)
\end{equation}
with $t\geq \frac{1-kt}{d-k}$. This allows us to compute the convex roofs of the Vidal monotones on isotropic states. %The details of the proof of the following theorem can be found in Appendix \ref{sec:isoentproof}.

\begin{theorem}\label{thm:vidaliso}
Consider the convex roof of the Vidal monotones $E_k$ on the isotropic states of $\CC^d\otimes\CC^d$. For $k\in\{1,\dots,d-1\}$ and $\wbeta\in[0,1]$, it holds that
   \begin{multline}\label{eq:vidalisot}
  \cohat{E_k}(\rho_{\isot}(\wbeta)) =\\
                          \left\{\begin{array}{ll}
                           0, &  \wbeta\in[0, \frac{k}{d}]\\
                          \tfrac{1}{d}\left(\sqrt{(1-\wbeta)k}-\sqrt{\wbeta(d-k)}\right)^2, & \wbeta\in[\frac{k}{d}, 1].
                         \end{array}\right.
 \end{multline}
 \end{theorem}
\begin{proof}
Note that $\cohat{E_k}(\rho_{\isot}(\wbeta))=\cohat{E_{k,{\isot}}}(\wbeta)$ by Theorem \ref{thm:coco}, where $E_{k,\isot}$ is the function as defined in Eq.~\eqref{eq:Eisotpsi} and the entanglement monotone used is $E=E_k$. An explicit form of \eqref{eq:Eisotpsi} for the Vidal monotones is computed in Eq.~\eqref{eq:ekisotformula} of Theorem~\ref{thm:isoentk} in Appendix \ref{sec:isoentproof}. It is clear that $E_{k,\isot}(\wbeta)$ in~\eqref{eq:ekisotformula} is convex as a function of $\wbeta$ (which may be confirmed by examining its second derivative). Thus $E_{k,\isot} = \cohat{E_{k,\isot}}$, which concludes the proof. 
\end{proof}

The convex roofs of the Vidal monotones can be trivially computed for $k\geq d$, in which case $E_{k,\isot}(\wbeta)=0$ for all $k\geq d$ and any $\wbeta$. A plot of the Vidal monotones~\eqref{eq:vidalisot} evaluated on isotropic states $\rho_{\isot}(\wbeta)$ with $d=5$ is shown in Fig\ \ref{fig:vidaliso}. 

It is perhaps interesting to note that the equation
\[
 y=\left(\sqrt{(1-x)\tfrac{k}{d}}-\sqrt{(1-\tfrac{k}{d})x}\right)^2
\]
is part of the unique ellipse that is tangent to the $x$-axis at the point $(\tfrac{k}{d},0)$, tangent to the $y$-axis at the point $(0,\tfrac{k}{d})$, and goes through the point $(1,1-\tfrac{k}{d})$. 

\begin{figure}
 \centering
 \includegraphics[width=.48\textwidth]{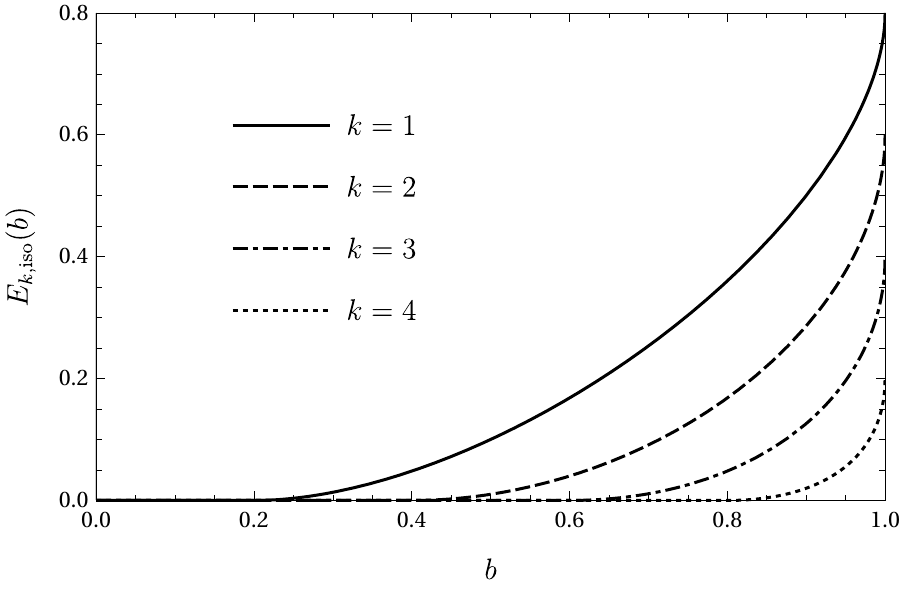}
 \caption{The convex roof of the Vidal monotones $E_1$, $E_2$, $E_3$, and $E_4$ evaluated on isotropic states with dimension $d=5$.}
 \label{fig:vidaliso}
\end{figure}

The resulting computations of computing the Vidal monotones on isotropic states can also be used to construct a lower bound for any arbitrary entanglement monotone evaluated on isotropic states. For any~$d\geq2$ and any $\wbeta\in[0,1]$, define the following Schmidt vector:
\begin{equation}\label{eq:lambdaalphaiso}
 \vec{\lambda}^\wbeta = \begin{pmatrix}
                         1-E_{1,\isot}(\wbeta)\\
                         E_{1,\isot}(\wbeta)-E_{2,\isot}(\wbeta)\\
                         \vdots\\
                         E_{d-2,\isot}(\wbeta)-E_{d-1,\isot}(\wbeta)\\
                         E_{d-1,\isot}(\wbeta)
                        \end{pmatrix}.
\end{equation}
%where we note that $E_{d,\isot}(\wbeta)=0$. 
For each $k$ it holds that $E_{k}(\vec{\lambda}^\wbeta)=E_{k,\isot}(\wbeta)$. By construction, it holds that $E_{k}(\psi)\geq E_{k}(\vec{\lambda}^\wbeta)$ for any pure state $\ket{\psi}$ that twirls to $\rho_{\isot}(\wbeta)$ (i.e.\ satisfying $\braketinner{\psi}{\Phi_d}{\psi}=\wbeta$). Thus $\vec{\lambda}\prec\vec{\lambda}^\wbeta$ where $\vec{\lambda}$ is a Schmidt vector of any pure state that twirls to $\rho_{\isot}(\wbeta)$. This implies that we can use $\vec{\lambda}^\wbeta$ to construct a lower bound for any entanglement monotone~$E$ evaluated on isotropic states. In particular, it holds that
\begin{equation}
 \cohat{E}(\rho_{\isot}(\wbeta))\geq E(\vec{\lambda}^\wbeta)
\end{equation}
for the convex roof of any possible entanglement monotone $E$ evaluated on isotropic states. 

%~~~~~~~~~~~~~~~~~~~~~~~~~~~~
\subsubsection{Generalized entropy measures}

It is also possible to study the convex roof of generalized measures of entropy, as studied in Ref.~\cite{Berry2003}, rather than entanglement measures. Generalized entropy measures are functions of the form $H_f(\vec{\lambda}) = \sum_{i}f(\lambda_i)$ for functions $f$ that satisfy the following conditions:
\begin{enumerate}[1.]
 \item $f(0)=0$
 \item $f$ is either strictly concave or strictly convex on the interval $[0,1]$ and 
 \item the first derivative $f'$ exists and is continuous on the interval $(0,1)$.
\end{enumerate} 
This includes the entropy of entanglement when $f(x)=-x\log x$, as well as quantities that are related to the Renyi entropies when $f(x)=x^\alpha$. In \cite{Berry2003} it was shown how to compute the minimum and maximum values of one generalized entropy $H_f(\vec{\lambda})$ for all Schmidt vectors $\vec{\lambda}$ with some other generalized entropy $H_g(\vec{\lambda})=c$ held constant. It turns out that the Schmidt vectors minimizing or maximizing these quantities will either be of the form
\begin{equation}\label{eq:veclamt1}
 \vec{\lambda} = \bigl(t,\tfrac{1-t}{d-1},\dots,\tfrac{1-t}{d-1}\bigr)
\end{equation}
where $t\geq\tfrac{1-t}{d-1}$, or
\begin{equation}\label{eq:veclamt2}
 \vec{\lambda} = \bigl(t,\dots,t,1-kt,0,\dots,0\bigr)
\end{equation}
where $t\geq 1-kt$, and there are $k= \lfloor\tfrac{1}{t}\rfloor$ probabilities equal to $t$. We can then make use of the following theorem.

\begin{theorem}[Theorem 1 in Ref.~\cite{Berry2003}]\label{thm:berry}
 Let $f\:[0,1]\rightarrow \RR$ and $g\:[0,1]\rightarrow\RR$ both satisfy conditions (i)-(iii) above.
 \begin{enumerate}
  \item If $f'\circ (g')^{-1}$ is strictly convex (concave), then the maximum (minimum) $H_f$ that can be achieved for fixed $H_g$ is obtained by a probability distribution of the form in Eq.~\eqref{eq:veclamt1}.
  \item  If $f'\circ (g')^{-1}$ is strictly convex (concave), then the minimum (maximum) $H_f$ that can be achieved for fixed $H_g$ is obtained by a probability distribution of the form in Eq.~\eqref{eq:veclamt2}. 
 \end{enumerate}
\end{theorem}

Note that $g$ in Theorem \ref{thm:berry} is either strictly concave or convex, so it must hold that $g'$ is in fact invertible on the interval~$(0,1)$. 

Given a function $f$ that satisfies the conditions above, we can define an entropy measure on pure states by $S_f(\psi)=H_f(\vec{\lambda})$, where $\vec{\lambda}$ here is the vector of Schmidt coefficients of $\ket{\psi}$. This can be extended to mixed states via the convex roof construction. Evaluating the convex roof of such an entropy measure on isotropic states $\rho_{\isot}(\wbeta)$ amounts to minimizing $H_f(\vec{\lambda})$ subject to the constraint $\sum_{i}\sqrt{\lambda_i} = \sqrt{d\wbeta}$. 
In particular we can evaluate functions of the form
\begin{equation}\label{eq:HfBerrytheorem}
 H_{f,\isot}(\wbeta) = \inf\biggl\{H_f(\vec{\lambda}) \,\bigg|\, \sum_{i=1}^d\sqrt{\lambda_i}=\sqrt{d\wbeta} \biggr\}
\end{equation}
for $\wbeta\in[\frac{1}{d},1]$. The constraint in Eq.~\eqref{eq:HfBerrytheorem} can be rewritten as $\sqrt{d\wbeta}=H_g(\vec{\lambda})$, where we choose $g(x)=\sqrt{x}$. If $f$ satisfies the conditions in Theorem \ref{thm:berry}, then we may use Theorem \ref{thm:berry} to compute the value in Eq.~\eqref{eq:HfBerrytheorem}. Note that $(g')^{-1}(x)=\frac{1}{4x^2}$, so it suffices to check if $f'(\frac{1}{4x^2})$ is either strictly concave or convex as a function of $x$. 

Using $\vec{\lambda}$ of the form in Eq.~\eqref{eq:veclamt1}, solving for $t$ with respect to the constraint $\sum_{i=1}^d\sqrt{\lambda_i}=\sqrt{d\wbeta}$ such that $H_f(\vec{\lambda})$ is minimized yields
\begin{equation}\label{eq:tgen1}
 t=1-\frac{1}{d}\left(\sqrt{1-\wbeta}-\sqrt{\wbeta(d-1)}\right)^2.
\end{equation}
Therefore, if $f'(\frac{1}{4x^2})$ is strictly concave, it follows that $H_{f,\isot}(\wbeta) = f(t) +(d-1)f(\frac{1-t}{d-1})$, where the value of $t$ is taken from Eq.~\eqref{eq:tgen1}. 

Using $\vec{\lambda}$ of the form in Eq.~\eqref{eq:veclamt2}, solving for $t$ with respect to the constraint $\sum_{i=1}^d\sqrt{\lambda_i}=\sqrt{d\wbeta}$ such that $H_f(\vec{\lambda})$ is minimized yields
\begin{equation}\label{eq:tgen2}
 t=\frac{\left(\sqrt{d\wbeta k}+\sqrt{k+1-d\wbeta}\right)^2}{k(k+1)^2},
\end{equation}
where $k=\lfloor d\wbeta\rfloor$. It follows that, if $f'(\frac{1}{4x^2})$ is strictly convex then $H_{f,\isot}(\wbeta) = \lfloor d\wbeta\rfloor f(t) +f(1-\lfloor d\wbeta\rfloor t)$, where the value of $t$ is taken from \eqref{eq:tgen2}. Example values of $t$ in Eqs.~\eqref{eq:tgen1} and \eqref{eq:tgen2} as functions of $\wbeta$ for $d=5$ are plotted in Fig\ \ref{fig:valuesoftopt}. 

\begin{figure}
 \includegraphics[width=.48\textwidth]{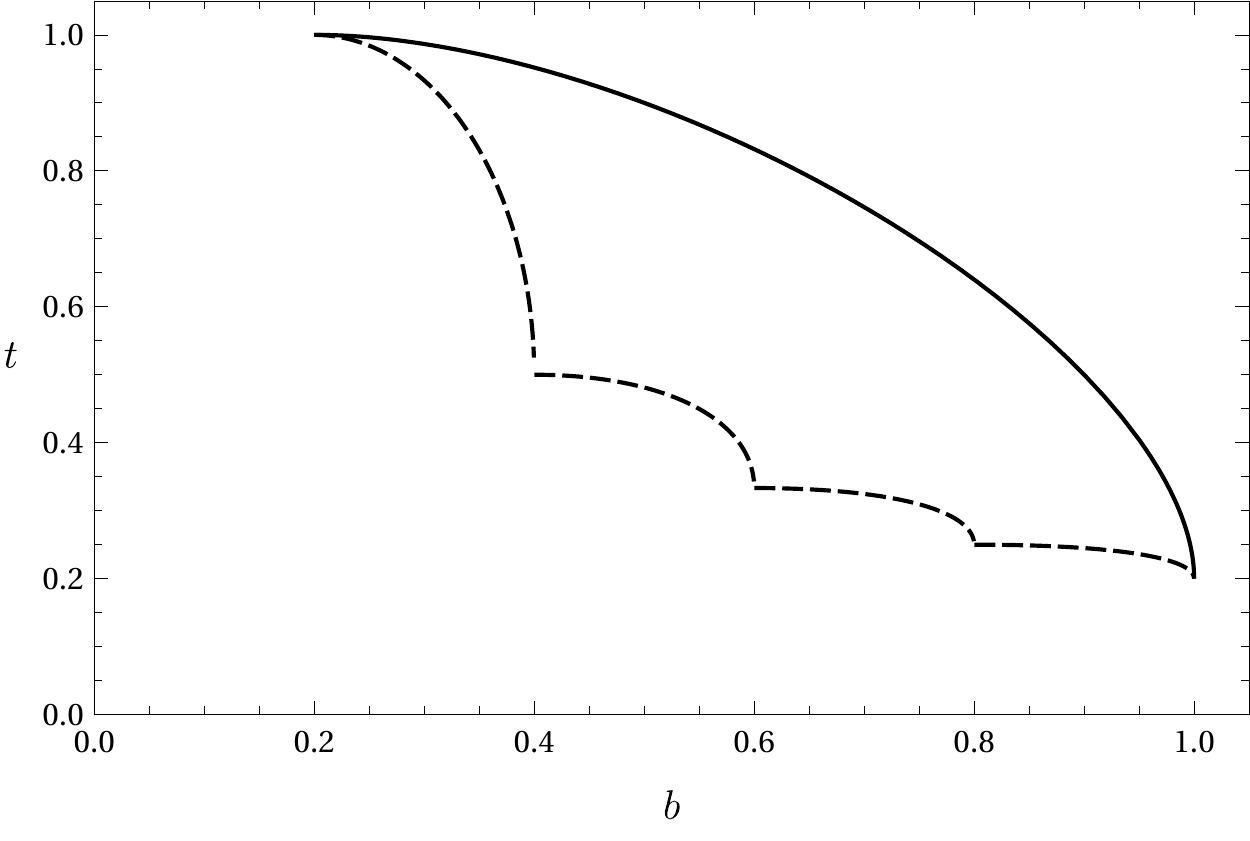}
 \caption{Example values of $t$ from Eqs.~\eqref{eq:tgen1} (solid line) and Eq.~\eqref{eq:tgen2} (dashed line) as functions of $\wbeta$ for $d=5$.}
 \label{fig:valuesoftopt}
\end{figure}

%~~~~~~~~~~~~~~~~~~~~~~~~~~~~
\subsubsection{Generalized concurrences}

Using the methods above, it is also possible to compute convex roofs of some of the \emph{generalized concurrence monotones} \cite{Gour2005}. These are defined as follows. For $k=1,2,\dots,d$,  let $S_k$ be the $k\textsuperscript{th}$ elementary symmetric polynomial of $d$ variables. That is,
\begin{align*}
 S_1(\vec{\lambda}) = \sum_{i=1}^d\lambda_i,\,\,\,  S_2(\vec{\lambda}) = \sum_{i<j}\lambda_i\lambda_j, \,\,
 \cdots, \,\,S_d(\vec{\lambda})=\prod_{i=1}^d\lambda_i.
\end{align*}
Note that $S_k(\frac{1}{d},\dots,\frac{1}{d})=\frac{1}{d^k}\binom{d}{k}$. The generalized concurrence monotones are defined by
\[
C_k(\vec{\lambda}) =\tfrac{d}{\binom{d}{k}^{1/k}}S_k(\vec{\lambda})^{1/k}.
\]
These symmetric functions are also concave \cite{Gour2005} and thus are valid entanglement monotones on pure states. Each $C_k$ achieves a maximum value of 1 on the maximally entangled pure state of two qudits. Note that $C_d$ is sometimes called the $G$-concurrence~\cite{Gour2005}.

Here, we compute the convex roofs of $C_2$ and $C_d$ on isotropic states. For $\wbeta\in[\tfrac{1}{d},1]$, we minimize $C_2$ and $C_d$ over all Schmidt vectors that satisfy $\sum_{i=1}^d\lambda_i=1$ and $\sum_{i=1}^d\sqrt{\lambda_i}=\sqrt{d\wbeta}$.   

We first compute $\widehat{C_2}$ for isotropic states. Note that 
\begin{align*}
 S_2(\vec{\lambda}) %&= \Bigl(\sum_{i=1}^d\lambda_i\Bigr)^2 - \sum_{i=1}^d \lambda_i^2\\
                    &=  \frac{1}{2}\Bigl(1- \sum_{i=1}^d \lambda_i^2\Bigr).
\end{align*}
Hence, minimizing $S_2(\vec{\lambda})$ is equivalent to maximizing $\sum_{i=1}^d \lambda_i^2$. By Theorem~\ref{thm:berry}, the optimal value  of this will be achieved by the Schmidt vector of the form in~\eqref{eq:veclamt1} with the value~$t$ from~\eqref{eq:tgen1}.  Thus
\begin{equation}\label{eq:C2t}
 C_{2,\isot}(\wbeta) = \frac{\sqrt{d}}{d-1}\sqrt{(1-t)(d(1+t)-2)},
\end{equation}
with $t$ from~\eqref{eq:tgen1} and $\wbeta\in[\tfrac{1}{d},1]$. The function in Eq.~\eqref{eq:C2t} is strictly concave as a function of $\wbeta$, thus its convex roof is the linear function
\begin{equation}
 \widehat{C_{2,\isot}}(\wbeta)=\left\{\begin{array}{ll}
                                  0, & 0\leq \wbeta\leq \tfrac{1}{d}\\
                                  \frac{d\wbeta-1}{d-1},& \tfrac{1}{d}\leq \wbeta\leq 1.
                                 \end{array}\right.
\end{equation}
Hence, the convex roof of the 2-concurrence on isotropic states reduces to $\widehat{C_2}(\rho_{\isot}(\wbeta))=\widehat{C_{2,\isot}}(\wbeta)$. This agrees with the result from Ref.~\cite{Eltschka2015}.

To compute the convex hull of the $G$-concurrence $\widehat{C_d}$ for isotropic states, note that 
\begin{align*}
\log S_d(\vec{\lambda}) &= \sum_{i=1}^d \log\lambda_i.
\end{align*}
Thus minimizing $S_d(\vec{\lambda})$ is equivalent to maximizing $\sum_{i=1}^d \log \lambda_i$. By Theorem~\ref{thm:berry}, the optimal value will be achieved by the Schmidt vector of the form in~\eqref{eq:veclamt2} with the value~$t$ from Eq.~\eqref{eq:tgen2}. Thus $C_{d,\isot}(\wbeta)=0$ for $\wbeta\leq 1-\frac{1}{d}$, and
\begin{equation}\label{eq:Ckt}
 C_{d,\isot}(\wbeta) = d\left(t^{d-1}-(d-1)t^d\right)^{1/d}
\end{equation}
for $\wbeta>1-\frac{1}{d}$, where
\[
 t=\frac{1}{d(d-1)}\left(\sqrt{(d-1)\wbeta}+\sqrt{1-\wbeta}\right)^2.
\]
The expression in Eq.~\eqref{eq:Ckt} is strictly concave as a function of $\wbeta$; thus its convex roof is just the linear function
\begin{equation}
 \widehat{C_{d,\isot}}(\wbeta)=\left\{\begin{array}{ll}
                                  0, & 0\leq \wbeta\leq 1- \tfrac{d}{d}\\
                                  d\wbeta -d+1,& 1-\tfrac{1}{d}\leq \wbeta\leq 1.
                                 \end{array}\right.
\end{equation}
Hence, the convex roof of the $G$-concurrence on isotropic states reduces to $\widehat{C_d}(\rho_{\isot}(\wbeta))=\widehat{C_{d,\isot}}(\wbeta)$. This agrees with the result from Ref.~\cite{Sentis2016}.

%~~~~~~~~~~~~~~~~~~~~~~~~~~~~~~~~~~~~~~~~~~~~~~~~~~~~~~~
\subsection{Extension to some non-symmetric states}
\label{sec:othersymms}

Here we show how to use the results from the previous sections to compute convex roof entanglement monotones for some states that are not necessarily symmetric.

For a subgroup $\calG\subset\LU$ of local unitaries and an entanglement monotone $E$ on pure states, recall that we can define the function
\[
E_{\calG}(\rho) := \min\{E(\psi)\,|\,\calT_\calG(\ketbra{\psi}{\psi})=\rho\}
\]
on $\calG$-invariant states $\rho$, where the minimization is taken over all pure states that twirl to $\rho$. A pure state $\ket{\psi}$ is said to minimize the entanglement of $\rho$ (with respect to $\calG$ and $E$) if $\calT_\calG(\ketbra{\psi}{\psi})=\rho$ and $E(\psi)=E_{\calG}(\rho)$. We also consider the orbit of $\ket{\psi}$ under the group $\calG$, which we denote as
\[
 \orb_\calG(\psi)= \left\{g\ketbra{\psi}{\psi}g^\dagger \,|\,g\in\calG\right\}.
\]

\begin{theorem}\label{thm:EGorb}
Let $\calG\subset\LU$ be a subgroup of local unitaries, let $\rho$ be a $\calG$-invariant state, and let $\ket{\psi}$ be a pure state that minimizes the entanglement of $\rho$ with respect to $E$ as defined in the preceding paragraph. If $\widehat{E_\calG}(\rho)=E_\calG(\rho)$, then 
\begin{equation}
 \widehat{E}(\sigma) = E(\psi) \quad \text{ for all } \sigma\in\co\left(\orb_\calG(\psi)\right),
\end{equation}
where $\co$ denotes the convex hull.
\end{theorem}
\begin{proof}
 Suppose the conditions of the theorem are satisfied and let $\sigma\in\co\left(\orb_\calG(\psi)\right)$. Since $\calG$ is a subgroup of local unitaries, it holds that $E(g\ket{\psi})=E(\psi)$ for all $g\in\calG$. It follows that $\widehat{E}(\sigma)\leq E(\psi)$ from the definition of the convex roof. Furthermore, since $\calT_\calG$ is an LOCC channel, it holds that $\widehat{E}(\calT_\calG(\sigma))\leq \widehat{E}(\sigma)$. Note that $\rho=\calT_\calG(\sigma)$ and $\widehat{E}(\rho)=\widehat{E_\calG}(\rho)$. The result follows.
\end{proof}

Theorem \ref{thm:EGorb} allows us to compute the convex hull on a larger class of non-symmetric states if we can find $\calG$-invariant states such that $\widehat{E_\calG}(\rho)=E_\calG(\rho)$. On the other hand, if $\widehat{E_\calG}(\rho)<E_\calG(\rho)$, we can still compute $\widehat{E}$ on a larger class of non-symmetric states under certain conditions. See Appendix \ref{sec:extensions} for details. 

%~~~~~~~~~~~~~~~~~~~~~~~~~~~
\subsection{Convex roofs on other symmetries}\label{sec:othersyms}

In Ref.~\cite{Vollbrecht2001}, it was shown how to extend the convex roof formula for the entanglement of formation $E_F$ from the Werner and isotropic states to a larger family of $OO$-invariant states. Here, we show that this can in fact be done for any entanglement monotone. Furthermore, we extend the convex roof formulas to the phase-permutation invariant states as well. 

Let $\calG$ and $\calH$ be subgroups of the local unitaries and $\calH\subset\calG$. The commutants of $\calG$ and $\calH$ satisfy $\calG'\subset\calH'$, so the family of $\calG$-invariant states forms a subset of the $\calH$-invariant states. If it is known how to compute the convex roofs of entanglement monotones on $\calG$-invariant states, then we can apply the result of Theorem \ref{thm:EGorb} to compute convex roofs on some $\calH$-invariant states that are also in the convex hull of the orbit of some minimizing pure state. That is, we can evaluate the convex roofs of entanglement monotones on states that are in the intersection
\[
 \calT_\calH(\D)\cap \co\left(\orb_\calG(\psi)\right)
\]
if $\ket{\psi}$ is a minimizing pure state for a $\calG$-invariant state, where $\orb_\calG(\psi)$ is the orbit of $\ket{\psi}$ is denoted by
\[
 \orb_\calG(\psi)=\{g\ketbra{\psi}{\psi}g^{-1}\,|\,g\in\calG\}.
\]

The minimizing pure states for Werner states are always the states $\ket{\psi_a}$ as defined in Eq.~\eqref{eq:psigam2}. We first show which of the phase-permutation Werner states $\rho_{\wer}^G(a,b)$ and $OO$-invariant states $\rho_{\Orth}(a,b)$ are in the orbits of these minimizing pure states. These are exactly the states depicted in regions A of Figs.~\ref{fig:ooinv1} and \ref{fig:pphaspermwer1}. This allows us to extend the formulas for convex roof entanglement monotones from the Werner states to this larger family of states. The proof of the following lemma can be found in Appendix~\ref{sec:extensions}.

\begin{lemma}\label{lem:coorbwer}
 Let $a\in[\frac{1}{2},1]$. Then
 \begin{enumerate}
  \item $\rho_{\wer}^G(a,b)\in\co(\orb_{\wer}(\psi_a))$ for all $b\in[0,1-a]$; and 
  \item $\rho_{\Orth}(a,b)\in\co(\orb_{\wer}(\psi_a))$ for all $b\in[0,\frac{2}{d}(1-a)]$. 
 \end{enumerate}
That is, all states in region \textup{A} of Fig~\ref{fig:ooinv1} and region \textup{A} of Fig~\ref{fig:pphaspermwer1} are in the convex hulls of the orbits of the corresponding minimizing pure states for $\rho_{\wer}(a)$.  
\end{lemma}

A similar statement can be made for isotropic states. Here, however, the form of the Schmidt coefficients of the minimizing pure state $\ket{\phi_b}=\sum_{i=1}^d\sqrt{\lambda_i}\ket{ii}$ for the isotropic state $\rho_{\isot}(b)$ will depend on which entanglement monotone $E$ is being considered. As above, the convex roof of $E$ can be evaluated
on any state in the convex hull of the orbit of $\ket{\phi_b}$. In the following lemma, we show which phase-permutation isotropic states and which $OO$-invariant states are in the convex hulls of these orbits. For any $E$, all phase-permutation isotropic states
$\rho_{\isot}(a,b)$ in region B of Fig.~\ref{fig:pphaspermiso1} are in the convex hull of the orbit of the minimizing pure state
$\ket{\phi_b}$. In most cases, all $OO$-invariant states $\rho_{\Orth}(a,b)$ in region B of Fig.~\ref{fig:ooinv1} are also in the
convex hull of the orbit of $\ket{\phi_b}$. The proof of the following lemma can be found in Appendix~\ref{sec:extensions}.

\begin{lemma}\label{lem:coorbiso}
 Let $E$ be an entanglement monotone on pure states and let $b\in[\frac{1}{d},1]$. Let $\ket{\phi_b}=\sum_{i=1}^d\sqrt{\lambda_i}\ket{ii}$ be the pure state that minimizes $E$ for $\rho_{\isot}(b)$. Then
 \begin{enumerate}
  \item $\rho_{\isot}^G(a,b)\in\co(\orb_{\isot}(\phi_b))$ for all $a\in[0,1-b]$, and
  \item If $\vec{\lambda}$ is of the form in either Eq.~\eqref{eq:lambdaisotoptmain} or Eq.~\eqref{eq:veclamt2}, then $\rho_{\Orth}(a,b)\in\co(\orb_{\isot}(\phi_b))$ for all $a\in[0,\frac{d(1-b)}{2(d-1)}]$. 
 \end{enumerate}
That is, all states in region \textup{B} of Fig~\ref{fig:ooinv1} and region \textup{B} of Fig~\ref{fig:pphaspermiso1} are in the convex hulls of the orbits of the corresponding minimizing pure states for $\rho_{\isot}(b)$.  
\end{lemma}

For every entanglement monotone considered in this work, the Schmidt coefficients of the minimizing pure states have this desired form. This allows us to extend the convex roofs of these entanglement monotones from the isotropic states to this larger family of states. 

If $E_{\wer}(a)$ and $E_{\isot}(b)$ are already convex as functions of $a$ and $b$, then Lemmas \ref{lem:coorbwer} and \ref{lem:coorbiso}, together with Theorem \ref{thm:EGorb}, allow us to extend these convex roof formulas to any state in regions \textup{A} of Figs.~\ref{fig:ooinv1} and \ref{fig:pphaspermwer1} and regions \textup{B} of Figs.~\ref{fig:ooinv1} and \ref{fig:pphaspermiso1}. It is noteworthy that the value of the convex roof for any entanglement monotone for these states depends only on one of the expectations $\Tr[\rho W_{\!\!-}]$ or $\Tr[\rho \Phi_d]$. As an example, the convex roofs of the Vidal monotones on the $OO$-invariant states with dimension $d=5$ are shown in Fig~\ref{fig:OO-vidals}. 
\begin{figure}
 \includegraphics[width=.9\columnwidth]{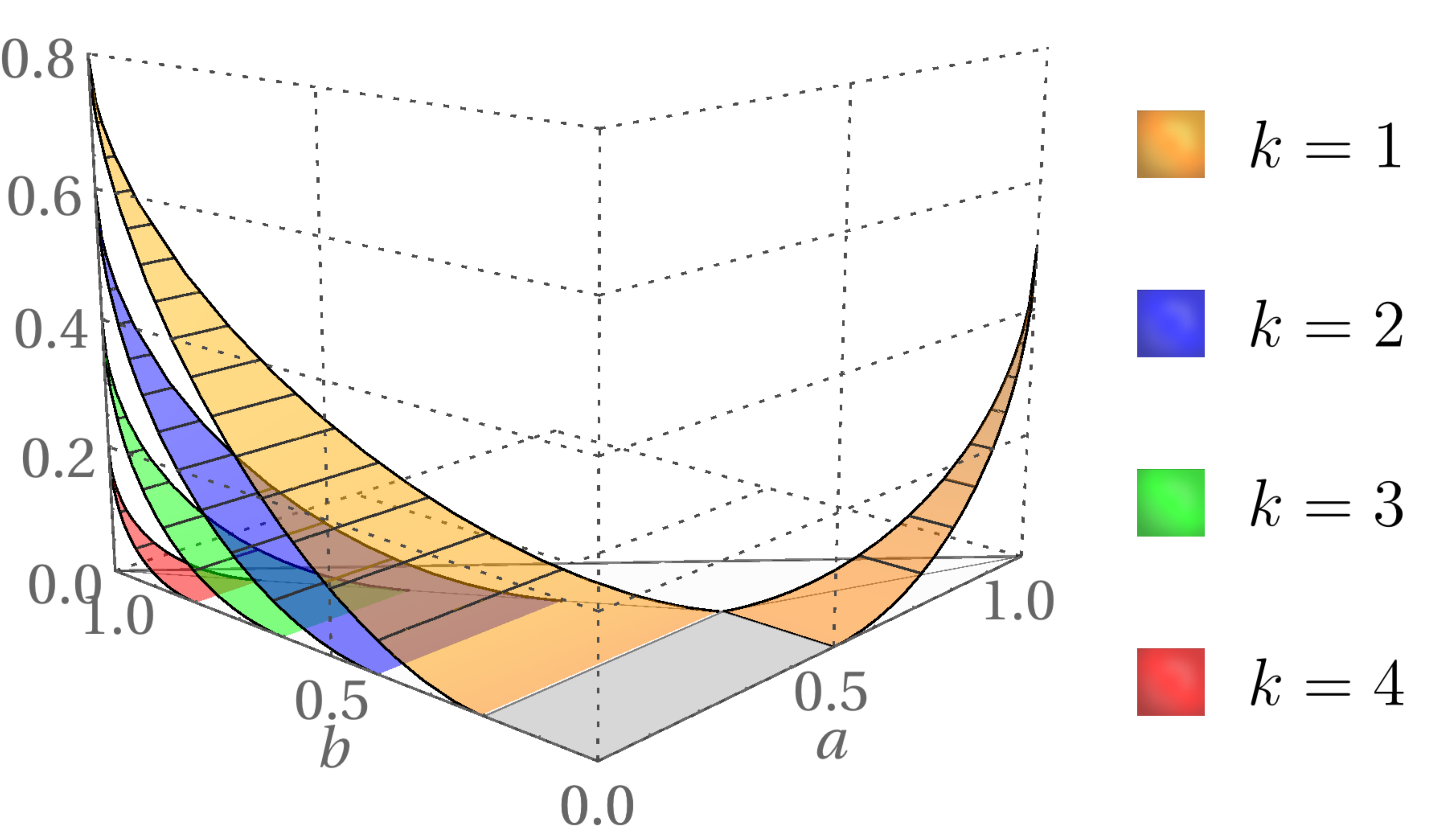}
 \caption{Convex roofs of the Vidal monotones $E_k$ (for $d=5$ and $k=1,2,3,4$) evaluated on regions A and B of the $OO$-invariant states. For the surfaces on the left-hand side, $k$ varies from 1 to 4 from the upper to the lower level. Only $E_1$ is non-vanishing on the right-hand side.}
 \label{fig:OO-vidals}
\end{figure}

If $E_{\wer}(a)$ and $E_{\isot}(b)$ are not convex, e.g.~if there is some value $b$ so that $\cohat{E_{\isot}}(b)<E_{\isot}(b)$, then we may still extend the formula to all of these states as long as $E$ is continuous.

%%%%%%%%%%%%%%%%%%%%%%%%%%%%%%%%%%%%%%%%%%%%%%%%%%%%%%%%%%%%

\section{Conversion witnesses}
\label{sec:witnesses}

It was shown in Ref.~\cite{Gheorghiu2008} that a pure state $\ket{\psi}\in\CC^d\otimes\CC^d$ can be converted into an arbitrary mixed state $\rho\in\D(\CC^d\otimes\CC^d)$ if and only if there exists an pure state decomposition $\{p_i,\ket{\varphi_i}\}$ of $\rho$ that satisfies
\[
 E_k(\psi)\geq \sum_{i}p_i E_k(\varphi_i)
\]
for all positive integers $k$, where $\rho=\sum_{i}p_i\ketbra{\varphi_i}{\varphi_i}$. This necessary and sufficient condition for LOCC transformation can be encoded into the following complete witness:
\begin{equation}\label{eq:Wconversion}
 W(\psi,\rho) = \max_{\{p_i,\ket{\varphi_i}\}} \min_{k} \biggl(E_k(\psi)-\sum_{i}p_iE_k(\varphi_i)\biggr),
\end{equation}
where the maximum is taken over all pure state decompositions $\rho=\sum_i p_i\ketbra{\varphi_i}{\varphi_i}$. The function $W$ is a complete witness in the sense that $W(\psi,\rho)\geq0$ if and only if $\ket{\psi}$ can be converted into $\rho$ via LOCC. Although this function cannot be computed for arbitrary mixed states, it can be simplified for certain classes of mixed states $\rho$. In particular, we compute $W(\psi,\rho)$ explicitly in the case when $\rho$ is a state on $\CC^2\otimes\CC^d$ for any $d$ (i.e.\ in the case when at least one subsystem is a qubit).  We can also make extensive use of symmetry to compute $ W(\psi,\rho)$ in the case when $\rho$ is highly symmetric (i.e.\ Werner or isotropic states).

The witness in Eq.~\eqref{eq:Wconversion} simplifies to a known necessary and sufficient condition for converting a pure state $\ket{\psi}$ to a mixed state $\rho$ in the case when $\rho$ is a state of a system in which one subsystem is a qubit \cite{Vidal2000b}. Indeed, for pure states $\ket{\varphi}\in\CC^2\otimes\CC^d$ with any $d\geq2$, it holds that $E_k(\varphi)=0$ whenever $k\geq2$ since  $\ket{\varphi}$ can have at most two nonzero Schmidt coefficients. If $\rho$ is any mixed state on $\CC^2\otimes\CC^d$, then the minimization over $k$ in Eq.~\eqref{eq:Wconversion} can be eliminated, since only $E_1$ can be nonzero. In this case, the conversion witness in Eq.~\eqref{eq:Wconversion} simplifies to $W(\psi,\rho) = E_1(\psi) - \cohat{E_1}(\rho)$. This implies the following theorem.

\begin{theorem}
 For any bipartite mixed state $\rho$ on $\CC^2\otimes\CC^d$ and for any bipartite pure state $\ket{\psi}$ of systems of any size, it holds that $\ket{\psi}\locc\rho$ if and only if $E_1(\psi)\geq\cohat{E_1}(\rho)$.
\end{theorem}

 Furthermore, it was shown in Ref.~\cite{Vidal2000b} that $\cohat{E_1}$ for an arbitrary mixed state of two qubits simplifies to 
 \[
   \cohat{E_1}(\rho) = \frac{1-\sqrt{1-C(\rho)^2}}{2},
 \] 
where $C(\rho)$ is the concurrence \cite{Wootters1998a} of $\rho$. Hence, a pure state $\ket{\psi}$ can be converted into a mixed state $\rho$ on $\CC^2\otimes\CC^2$ if and only if $C(\psi)\geq C(\rho)$.

As the following theorem shows, the value of $\cohat{E_1}$ gives a necessary and sufficient condition for converting any pure states into Werner states of arbitrary dimension as well. 

\begin{theorem}
 For any bipartite pure state $\ket{\psi}$ and any $\walpha\in[\frac{1}{2},1]$, it holds that $\ket{\psi}\locc\rho_{\wer}(\walpha)$ if and only if $\lambda_1\leq \frac{1}{2}+\sqrt{\walpha(1-\walpha)}$, where $\lambda_1$ is the largest Schmidt coefficient of $\ket{\psi}$.
\end{theorem}

Note that if $a\in[0,\frac{1}{2}]$ then $\rho_{\wer}(\walpha)$ is separable and thus $\ket{\psi}\locc\rho_{\wer}(\walpha)$ holds trivially. The theorem states the conditions for conversion in the case when $\rho_{\wer}(\walpha)$ is entangled.

\begin{proof}
 Let $a\in[\frac{1}{2},1]$ and suppose that $\ket{\psi}\locc\rho_{\wer}(\walpha)$. Then it must be the case that $E_1(\psi)\geq \cohat{E_1}(\rho_{\wer}(\walpha))$ since $E_1$ is an entanglement monotone. The result follows, since $E_1(\psi)=1-\lambda_1$ and $\cohat{E_1}(\rho_{\wer}(\walpha))=\frac{1}{2}-\sqrt{\walpha(1-\walpha)}$. 
 On the other hand, if $\lambda_1\leq \frac{1}{2}+\sqrt{\walpha(1-\walpha)}$ then $\vec{\lambda}\prec\vec{\lambda}^\walpha$, where $\vec{\lambda}$ is the vector of Schmidt coefficients of $\ket{\psi}$ and $\vec{\lambda}^\walpha$ is the vector of Schmidt coefficients of $\ket{\psi_\walpha}$ given in Eq.~\eqref{eq:psigam2}. It follows that $\ket{\psi}$ can be converted into $\ket{\psi_\walpha}$ by LOCC, but $\ket{\psi_\walpha}$ can be converted into $\rho_{\wer}(\walpha)$ via LOCC, since $\calT_{\wer}(\ketbra{\psi_\walpha}{\psi_\walpha})=\rho_{\wer}(\walpha)$ and the twirling operation $\calT_{\wer}$ is LOCC. This concludes the proof.
\end{proof}

We have shown that the conversion witness in Eq.~\eqref{eq:Wconversion} can be computed explicitly in the cases when $\rho$ is a Werner state or any state on a $\CC^2\otimes\CC^d$ system, but it remains unknown if it can be computed explicitly for any other classes of states. However, it may still be useful to consider upper and lower bounds of this quantity, since these would give either necessary or sufficient conditions for LOCC conversion from $\ket{\psi}$ into $\rho$. In particular, in the case when $\rho=\rho_{\isot}(\wbeta)$ is an isotropic state, a lower bound for \eqref{eq:Wconversion} can be found. The following theorem gives a no-go conversion witness for detecting when pure states cannot be converted into isotropic states. 

\begin{theorem}
 Let $\ket{\psi}$ be a pure state and $\wbeta\in[\tfrac{1}{d},1]$. If $\ket{\psi}\locc\rho_{\isot}(\wbeta)$ then $W_{\isot}(\vec{\lambda},\wbeta)\geq 0$, where 
\begin{equation}\label{eq:witness1}
 W_{\isot}(\vec{\lambda},\wbeta) = \max_{\vec{\mu}} \min_k \bigl(E_k(\vec{\lambda})-E_k(\vec{\mu})\bigr)
\end{equation}
and the the maximum is taken over all Schmidt vectors~$\vec{\mu}$ that satisfy $\sum_{i}\sqrt{\mu_i} = \sqrt{d\wbeta}$.  
\end{theorem}

\begin{proof}
 In the case when $\rho=\rho_{\isot}(\wbeta)$, it is clear that a lower bound for the witness $W$ in Eq.~\eqref{eq:Wconversion} can be given by
 \begin{equation}\label{eq:wpsiiso}
 W(\psi,\rho_{\isot}(\wbeta)) \geq \max_{\ket{\varphi}} \min_k \bigl(E_k(\psi)-E_k(\varphi)\bigr),
\end{equation}
where the maximum is taken over all $\ket{\varphi}\in\CC^d\otimes\CC^d$ such that $\braketinner{\varphi}{\Phi_d}{\varphi}=\wbeta$.  The left-hand side of the inequality in~\eqref{eq:wpsiiso} can be further simplified to the desired expression in Eq.~\eqref{eq:witness1}.
\end{proof}

In particular, if $W_{\isot}(\psi,\wbeta)<0$ then $\ket{\psi}\centernot\locc\rho_{\isot}(\wbeta)$. Although the formula for this witness is now much simpler than the general one in Eq.~\eqref{eq:Wconversion}, it still cannot be computed analytically for arbitrary Schmidt vectors $\vec{\lambda}$. However, we present a way to numerically compute these witnesses efficiently in Appendix \ref{sec:conversionwitnesscomputation}. 
\begin{figure}[t!]
  \includegraphics[width=\columnwidth]{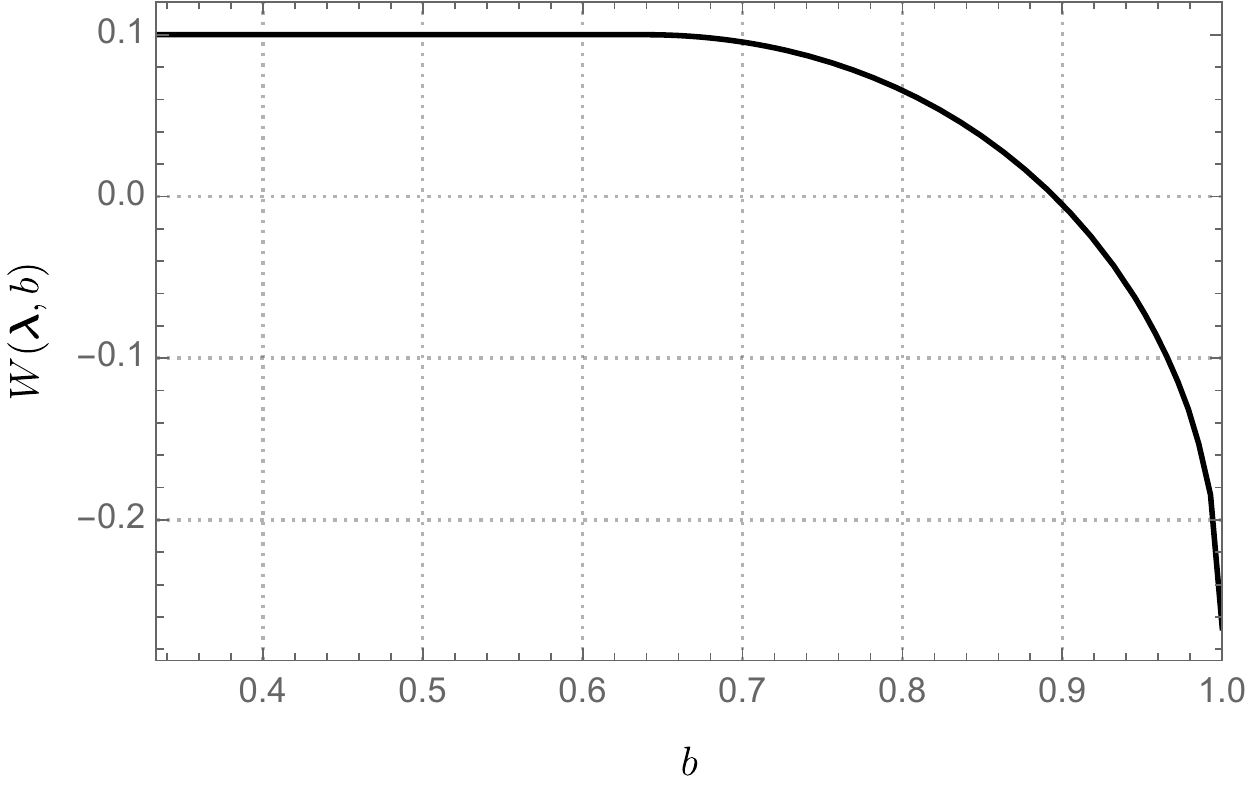}
 \caption{An example of the witness in Eq.~\eqref{eq:witness1} computed for $d=3$ and $\vec{\lambda}=(\frac{6}{10},\frac{3}{10},\frac{1}{10})$. It appears that $W(\vec{\lambda},b)<0$ whenever $b>0.895$. }\label{fig:witplot}
\end{figure}
An example of the witness in Eq.~\eqref{eq:witness1} with $d=3$ and $\vec{\lambda}=(\frac{6}{10},\frac{3}{10},\frac{1}{10})$ is shown in Fig\ \ref{fig:witplot}.  In this case, it appears that $W(\vec{\lambda},b)<0$ whenever $b>0.895$.  Hence the conversion $\ket{\psi}\locc\rho_{\isot}(\wbeta)$ is not possible when $b>0.895$, where $\ket{\psi}$ is the pure state with Schmidt coefficients $\vec{\lambda}$.\\

%%%%%%%%%%%%%%%%%%%%%%%%%%%%%%%%%%%%%%%%%%%%%%%%%%%%%%%%%%%%
\section{Conclusion}
\label{sec:conclude}

We computed the convex roof of entanglement monotones on certain classes of symmetric states. This generalized the work of Refs.~\cite{Vollbrecht2001} and \cite{Terhal2000}, where the entanglement of formation was computed for Werner and isotropic states. In particular, we computed the convex roof for any entanglement monotone on Werner states. The convex roof of certain types of monotones was also computed on isotropic states. We were able to extend these formula for the convex roofs to many non-symmetric states as well. In particular, for many states with other types of symmetries (i.e., for $OO$-invariant states as well as phase-permutation Werner and isotropic type states), we were also able to compute the convex roofs of these monotones.

We also constructed a necessary and sufficient condition in the form of a conversion witness that determines when a bipartite pure state can be converted to any Werner state by LOCC. A similar conversion witness was constructed for detecting when a pure state can be converted into an isotropic state, but the condition was only necessary and not sufficient.

This work sheds light on the structure of bipartite entanglement of symmetric states, an area of research that is still quite active. Recently, work has been done on computing convex roofs of certain entanglement monotones on larger classes of symmetric states \cite{Sentis2016a}. Investigations into further types of symmetries and other entanglement monotones will prove fruitful in the complete characterization of the LOCC convertibility of bipartite quantum entanglement.

\begin{acknowledgments}
The authors are grateful for fruitful discussions with Gael Sent\'is and Jens Siewert that took place during the workshop on multipartite entanglement at the Centro de Ciencias de Benasque Pedro Pascual in May 2016. M.G. is supported by an Izaak Walton Killam Memorial Scholarship from the Killam Trusts(Canada) and a Graduate Student Scholarship from Alberta Innovates–Technology Futures (Canada).
\end{acknowledgments}

%%%%%%%%%%%%%%%%%%%%%%%%%%%%%%%%%%%%%%%%%%%%%%%%%%%%%%%%%%%%
\bibliographystyle{myapsrevbib}
\bibliography{/home/mark/Documents/library.bib}

%%%%%%%%%%%%%%%%%%%%%%%%%%%%%%%%%%%%%%%%%%%%%%%%%%%%%%%%%%%%
\appendix

%%%%%%%%%%%%%%%%%%%%%%%%%%%%%%%%%%%%%%%%%%%%%%%%%%%%%%%%%%%%
\section{Convex roof of entanglement monotones for Werner states}
In this section, we present the proof of Theorem \ref{lem:Ewerequal}, which follows the ideas for computing the entanglement of formation for Werner states as presented in Ref.~\cite{Vollbrecht2001}.

\label{sec:werentproof}
\begin{proof}[Proof (of theorem \ref{lem:Ewerequal})]
 If $\walpha\in[0,\frac{1}{2}]$ then $E(\psi_\walpha)=0$ since $\ket{\psi_{\walpha}}$ is separable, so the conclusion is trivially true. Suppose that $\walpha\in[\frac{1}{2},1]$ and let $\ket{\psi}$ be another pure state  satisfying $\braketinner{\psi}{W_{\!\!-}}{\psi}= \walpha$. Let $\vec{\lambda},\vec{\lambda}^\walpha\in\RR^d$ denote the Schmidt vectors of $\ket{\psi}$ and $\ket{\psi_\walpha}$ respectively. We will show that $\vec{\lambda}\prec\vec{\lambda}^{\walpha}$. Since 
 \[
  \vec{\lambda}^{\walpha} = \left(\tfrac{1}{2}+\sqrt{\walpha(1-\walpha)}, \tfrac{1}{2}-\sqrt{\walpha(1-\walpha)},0,\dots,0\right)
 \]
has only two nonzero elements, it suffices to show that $\max(\vec{\lambda})\leq \tfrac{1}{2}+\sqrt{\walpha(1-\walpha)}$.

Without loss of generality we may suppose that $\ket{\psi}$ is of the form
\begin{equation}\label{eq:psiform}
 \ket{\psi} = U\otimes I \sum_{i=1}^d \sqrt{\lambda_i} \,\ket{ii}
\end{equation}
for some unitary operator $U$. Then
\begin{align*}
 \walpha = \bra{\psi}W_{\!\!-}\ket{\psi}
 &=\frac{1}{2}\biggl(1 - \sum_{i,j=1}^d\sqrt{\lambda_i\lambda_j}\braketinner{i}{U}{j} \braketinner{i}{U^\dagger }{j}\biggr)\\
 & = \frac{1}{2}\biggl(1 - \sum_{i,j=1}^d\sqrt{\lambda_i\lambda_j} U_{ij}\overline{U}_{ji}\biggr)\\
 &=\frac{1}{4}\sum_{i,j=1}^d \abs{\sqrt{\lambda_i}U_{ij}-\sqrt{\lambda_j}U_{ji}}^2,
\end{align*}
where $U_{ij}=\bra{i}U\ket{j}$ are the matrix elements of $U$. Since~$U$ is unitary, it holds that $\sum_{j}\abs{U_{ij}}^2=1$ for each~$i$ and thus $\sum_{i,j}\lambda_i\abs{U_{ij}}^2=1$. For each $i,j\in\{1,\dots,d\}$, define the probabilities
\begin{equation*}
 p_{ij} = \frac{\lambda_{i}\abs{U_{ij}}^2+\lambda_{j}\abs{U_{ji}}^2}{2} 
\end{equation*}
such that $p_{ij}\geq 0$ and $\sum_{i,j}p_{ij}=1$. Note that $p_{ij}=p_{ji}$. For all $i$ and $j$ such that $p_{ij}\neq 0$, define the quantities
\begin{equation*}
 z_{ij} = \frac{\sqrt{\lambda_{i}}U_{ij}}{\sqrt{\lambda_{i}\abs{U_{ij}}^2+\lambda_{j}\abs{U_{ji}}^2}}
\quad\text{and}\quad
\walpha_{ij}  = \frac{\abs{z_{ij}-z_{ji}}^2}{2}
\end{equation*}
such that $\abs{z_{ij}}^2+\abs{z_{ji}}^2=1$ and $\walpha_{ij}\in[0,1]$. Define the Schmidt vectors 
\begin{equation*}\label{eq:muij}
 \vec{\mu}^{(ij)} = \abs{z_{ij}}^2\vecbf{e}_i + \abs{z_{ji}}^2\vecbf{e}_j
\end{equation*}
  where $\{\vecbf{e}_1,\dots,\vecbf{e}_d\}$ are the standard basis vectors of $\RR^d$. It follows that
\[
 \sum_{i,j=1}^dp_{ij}\walpha_{ij} = \walpha \qquad\text{and}\qquad \sum_{i,j=1}^d p_{ij}\vec{\mu}^{(ij)}= \vec{\lambda}.
\]
That is, the quantity $\walpha$ and the Schmidt vector $\vec{\lambda}$ can be written as convex combinations of quantities $\walpha_{ij}\in[0,1]$ and Schmidt vectors $\vec{\mu}^{(ij)}$ using the same weights~$p_{ij}$.
Since $\abs{z_{ij}}^2 + \abs{z_{ji}}^2=1$ and $\abs{\abs{z_{ij}}^2-\abs{z_{ji}}^2}\leq \abs{z_{ij}{}^2-z_{ji}{}^2}$, we see that
\begin{align}
2 \max(\vec{\mu}^{(ij)}) 
    & = 2 \max\{\abs{z_{ij}}^2,\abs{z_{ji}}^2\}\nonumber\\
    &= \abs{z_{ij}}^2 + \abs{z_{ji}}^2 + \Bigl\lvert\abs{z_{ij}}^2-\abs{z_{ji}}^2\Bigr\rvert\nonumber\\
    &\leq 1+\abs{{z_{ij}}^2-{z_{ji}}^2}\nonumber\\
    &= 1+\abs{z_{ij}-z_{ji}}\abs{z_{ij}+z_{ji}}\label{eq:alphaijzij}.
\end{align}
Furthermore note that $1-\walpha_{ij}=\frac{1}{2}\abs{z_{ij}+z_{ji}}^2$, and thus
\begin{equation*}
\abs{z_{ij}-z_{ji}}\abs{z_{ij}+z_{ji}} =  2\sqrt{\walpha_{ij}(1-\walpha_{ij})} .
\end{equation*}
From Eq.~\eqref{eq:alphaijzij} it follows that $\max(\vec{\mu}^{(ij)})\leq \tfrac{1}{2}+\sqrt{\walpha_{ij}(1-\walpha_{ij})}$. Since $\max$ is a convex function on $\RR^d$, it follows that
\begin{align}
\max(\vec{\lambda})  
         &=  \max\biggl(\sum_{i,j=1}^dp_{ij}\vec{\mu}^{(ij)}\biggr)\nonumber\\
         &\leq \sum_{i,j=1}^d p_{ij}\max(\vec{\mu}^{(ij)})\nonumber\\
         &\leq \sum_{i,j=1}^d p_{ij}\left(\tfrac{1}{2}+\sqrt{\walpha_{ij}(1-\walpha_{ij})}\right)\nonumber\\
         &\leq \tfrac{1}{2}+\sqrt{\walpha(1-\walpha)},\label{eq:maxveclamFinalineq}
\end{align}
where the final inequality in Eq.~\eqref{eq:maxveclamFinalineq} follows from the concavity of the function $f(t) = \sqrt{t(1-t)}$. This yields the desired result that $\vec{\lambda}\prec\vec{\lambda}^\walpha$.

By Nielsen's majorization theorem \cite{Nielsen1999}, it follows that $E(\psi)\geq E(\psi_\walpha)$. 
\end{proof}

%%%%%%%%%%%%%%%%%%%%%%%%%%%%%%%%%%%%%%%%%%%%%%%%%%%%%%%%%%%%
\section{Convex roof of Vidal monotones for isotropic states}
\label{sec:isoentproof}
In this section, we present the proof of Lemma \ref{lem:Eisoequal} and supply the details for the proof of Theorem \ref{thm:vidaliso}. The proof of Lemma \ref{lem:Eisoequal} follows directly from the following Lemma.

\begin{lemma}\label{lem:lem4prec}
 Let $\wbeta\in[\frac{1}{d},1]$ and let $\ket{\psi}$ be a pure state with Schmidt vector $\vec{\lambda}$ satisfying $\braketinner{\psi}{\Phi_d}{\psi}=\wbeta$. There exists a pure state $\ket{\psi'}=\sum_{i}\sqrt{\lambda_i'}\ket{ii}$ such that
 \[
  \braketinner{\psi'}{\Phi_d}{\psi'} = \frac{1}{d}\Bigl(\sum_{i=1}^d\sqrt{\lambda_i'}\Bigr)^2 =\wbeta= \braketinner{\psi}{\Phi_d}{\psi}
 \]
and $\vec{\lambda}'\succ \vec{\lambda}$, where $\vec{\lambda}'$ is the Schmidt vector for $\ket{\psi'}$. 
\end{lemma}
\begin{proof}
We may suppose without loss of generality that $\ket{\psi}$ is of the same form as Eq.~\eqref{eq:psiform}. Thus
 \begin{align}
  \wbeta = \braketinner{\psi}{\Phi_d}{\psi}
  %&= \frac{1}{d}\biggl|\sum_{i=1}^d\sqrt{\lambda_i}\braketinner{i}{U}{i}\biggr|^2\nonumber\\
  &=\frac{1}{d}\biggl|\sum_{i=1}^d\sqrt{\lambda_i}U_{ii}\biggr|^2\nonumber\\
  &\leq\frac{1}{d}\Bigl(\sum_{i=1}^d\sqrt{\lambda_i}\Bigr)^2,\label{eq:alphaisotbd}
 \end{align}
where we note that $\abs{U_{ii}}\leq 1$ for all $i$ since $U$ is unitary. If $\sum_{i}\sqrt{\lambda_i}=\sqrt{d\wbeta}$ then we may set $\vec{\lambda}'=\vec{\lambda}$ and we are done. Suppose instead that the inequality in Eq.~\eqref{eq:alphaisotbd} is strict. Define a continuous function~\mbox{$s\colon\RR^d\rightarrow \RR$},
\begin{equation}\label{eq:sequation}
 s(\vec{\lambda})=\frac{1}{d}\Bigl(\sum_{i=1}^d\sqrt{\lambda_i}\Bigr)^2.
\end{equation}
We may suppose that the entries of $\vec{\lambda}$ are decreasing. For all $p\in[0,1]$ define the Schmidt vectors
\begin{align*}
 \vec{\lambda}'(p)
 &=(1-p)\vec{\lambda} + p (1,0,\dots,0).
\end{align*}
Note that $s(\vec{\lambda}'(p))$ is continuous and strictly decreasing as a function of $p$ and that 
\[
 \frac{1}{d} = s(\vec{\lambda}'(1)) < \wbeta < s(\vec{\lambda}'(0)) = s(\vec{\lambda}).
\]
By continuity of $s$, there exists a $p\in(0,1)$ such that $s(\vec{\lambda}'(p))=\wbeta$. Finally we note that $\vec{\lambda}'(p)\succ \vec{\lambda}$ for all $p$, which concludes the proof.
\end{proof}

We now supply the proof of Lemma \ref{lem:Eisoequal}. For an entanglement monotone $E$, recall that $E_{\isot}$ is defined as
\[
 E_{\isot}(\wbeta)=\min\bigl\{E(\vec{\psi})\,\big|\, \braketinner{\psi}{\Phi_d}{\psi}=\wbeta\bigr\}.
\]

\begin{proof}[Proof \textup{(}of Lemma \ref{lem:Eisoequal}\textup{)}]
First consider the case $\wbeta\in[\frac{1}{d},1]$. For all pure states $\ket{\psi}$ satisfying $\braketinner{\psi}{\Phi_d}{\psi}=\wbeta$, from Lemma \ref{lem:lem4prec} we can find a pure state $\ket{\psi'}$ with Schmidt coefficients $\vec{\lambda}'$ satisfying $\braketinner{\psi'}{\Phi_d}{\psi'}=s(\vec{\lambda}')=\wbeta$
with $\vec{\lambda}\prec\vec{\lambda}'$. It follows that $E(\vec{\lambda}')=E(\psi')\leq E(\psi)$. Hence we may restrict the minimization in Eq.~\eqref{eq:Eisotpsi} to states of the form $\ket{\psi}=\sum_{i}\sqrt{\lambda_i}\ket{ii}$. This implies that the computation of $E_{\isot}(\wbeta)$ may be simplified to
 \begin{align*}
  E_{\isot}(\wbeta)
  &=\min\Bigl\{E(\vec{\psi})\,\big|\, \ket{\psi}=\sum_{i=1}^d\sqrt{\lambda_i}\ket{ii}\text{ and }s(\vec{\lambda})=a\Bigr\}\\
  &=\min\Bigl\{E(\vec{\lambda})\,\Big|\, \sum_{i=1}^d\sqrt{\lambda_i} = \sqrt{d\wbeta}\Bigr\}
 \end{align*}
as desired.

Last we consider the case when $\wbeta\in[0,\frac{1}{d}]$. Consider the pure state 
\[
 \ket{\psi} = \sqrt{d\wbeta}\,\ket{11}+\sqrt{1-d\wbeta}\,\ket{12}.
\]
Then $\braketinner{\psi}{\Phi_d}{\psi}=\wbeta$, but $E(\psi)=0$ since $\ket{\psi}$ is separable. It follows that $E_{\isot}(\wbeta)=0$. This concludes the proof.
\end{proof}

We now proceed to compute the convex roof of the Vidal monotones for isotropic states. The following lemma shows that $E_k$ vanishes on the isotropic states with $\wbeta\in[0,\frac{k}{d}]$.

\begin{lemma}\label{lem:tkd1}
 For any integer $1\leq k\leq d$, it holds that $E_{k,\isot}(\wbeta)=0$ for all $\wbeta\in[0,\tfrac{k}{d}]$.
\end{lemma}
\begin{proof}
Since $E_k$ is an entanglement monotone on pure states, the result of Lemma \ref{lem:Eisoequal} shows that $E_{k,\isot}(\wbeta)=0$ whenever $\wbeta\in[0,\tfrac{1}{d}]$. So we may suppose that $k\geq 2$ and $\wbeta\in[\tfrac{1}{d},\frac{k}{d}]$. Consider the function $s$ defined in Eq.~\eqref{eq:sequation} restricted to the subset of Schmidt vectors $\vec{\lambda}$ that have at most $k$ nonzero entries. The function~$s$ achieves the values $\frac{1}{d}$ and~$\frac{k}{d}$ on this restriction, since
\[
s\big((1,0,\dots,0)\big) = \tfrac{1}{d} 
\quad\text{and}\quad
s\left(\big(\tfrac{1}{k},\dots,\tfrac{1}{k},0,\dots,0)\right) = \tfrac{k}{d} .
\]
The subset of Schmidt vectors in $\RR^d$ containing at most~$k$ nonzero elements is also connected. By continuity of~$s$, for any intermediate value $\wbeta\in[\frac{1}{d},\frac{k}{d})$ there exists a Schmidt vector $\vec{\lambda}$ with at most $k$ nonzero entries satisfying $s(\vec{\lambda})=\wbeta$. Since $E_k(\vec{\lambda})=0$ for all such $\vec{\lambda}$, it follows that $E_{k,\isot}(\wbeta)=0$ whenever $\tfrac{1}{d}\leq\wbeta\leq \tfrac{k}{d}$. 
\end{proof}

\begin{theorem}\label{thm:isoentk}
 Let $k\geq1$ be an integer. It holds that
   \begin{multline}\label{eq:ekisotformula}
  E_{k,\isot}(\wbeta) =\\\left\{\begin{array}{ll}
                           0, &  \wbeta\in [0,\frac{k}{d}]\\
                           \tfrac{1}{d}\Bigl(\sqrt{(1-\wbeta)k}-\sqrt{\wbeta(d-k)}\Bigr)^2, & \wbeta\in[\frac{k}{d}, 1].
                          \end{array}\right.
 \end{multline}
\end{theorem}

\begin{proof}
It was shown in Lemma \ref{lem:tkd1} that $E_{k,\isot}(\wbeta)=0$ whenever $\wbeta\in [0,\frac{k}{d}]$, so it remains to compute $E_{k,\isot}(\wbeta)$ when $\wbeta\in[\frac{k}{d}, 1]$. Computing $E_{k,\isot}(\wbeta)$ may be restated as the following optimization problem:
 \begin{align*}
  \text{maximize: }  & \lambda_1+\cdots+\lambda_k\\
  \text{subject to: }& \sum_{i=1}^d\lambda_i=1\text{ and } \sum_{i=1}^d\sqrt{\lambda_i}= \sqrt{d\wbeta}.
 \end{align*}
It is not difficult to see (by using Lagrange multipliers) that the optimal $\vec{\lambda}$ must be of the form
\begin{equation}\label{eq:lambdaisotopt}
 \vec{\lambda} = \bigl(\underbrace{t,\dots,t}_{k},\underbrace{\tfrac{1-kt}{d-k},\dots,\tfrac{1-kt}{d-k}}_{d-k}\bigr)
\end{equation}
for some $t\in[\frac{1}{d},\frac{1}{k}]$. For $\vec{\lambda}$ of this form, we see that
\begin{align*}
 \sum_{i=1}^d\sqrt{\lambda_i}
 &= k\sqrt{t}+(d-k)\sqrt{\tfrac{1-kt}{d-k}}\nonumber\\
 &= k\sqrt{(1-t)}+\sqrt{(d-k)(1-kt)}.
\end{align*}
For $\wbeta\in[\frac{k}{d},1]$, the largest positive value of $t$ that satisfies $k\sqrt{(1-t)}+\sqrt{(d-k)(1-kt)} = \sqrt{d\wbeta}$
is given by 
\begin{equation}\label{eq:tisotopt}
t = \frac{1}{k}-\frac{1}{kd}\bigl(\sqrt{(1-\wbeta)k}-\sqrt{\wbeta(d-k)}\bigr)^2.
 \end{equation}
 For $\vec{\lambda}$ as given in Eq.~\eqref{eq:lambdaisotopt} with $t$ as in Eq.~\eqref{eq:tisotopt}, it follows that
\begin{align*}
 E_{k,\isot}(\vec{\lambda}) 
 &= 1-(\lambda_1+\cdots+\lambda_k) \\
 &=1- kt\\
 &= \tfrac{1}{d}\bigl(\sqrt{(1-\wbeta)k}-\sqrt{\wbeta(d-k)}\bigr)^2,
\end{align*}
as desired.
\end{proof}

%%%%%%%%%%%%%%%%%%%%%%%%%%%%%%%%%%%%%%%%%%%%%%%%%%%%%%%%%%%%
\section{Convex roofs on further symmetric states}
\label{sec:extensions}
%-----------------------------------------------------------
\subsection{Proof of Lemma \ref{lem:coorbwer}}

\begin{proof}[Proof \textup{(}of Lemma \ref{lem:coorbwer} part 1\textup{)}]
By convexity, it suffices to check only the states on the boundary. That is, we check $\rho_{\wer}^G(\walpha,\wbeta)$ with $\wbeta=0$ and $\wbeta=1-\walpha$. In both cases, we find a pure state $\ket{\psi}\in\orb_{\wer}(\psi_\walpha)$ such that $\calT_{\wer}^G(\psi)=\rho_{\wer}^G(\walpha,\wbeta)$. 
 
Note that $\bra{\psi_\walpha}Q\ket{\psi_\walpha} = 0$. Thus $\calT_{\wer}^G(\psi_\walpha)=\rho_{\wer}^G(a,0)$ and thus $\rho_{\wer}^G(a,0)\in\co(\orb_{\wer}(\psi_a))$. For $\rho_{\wer}^G(a,1-a)$, consider the unitary block matrix
\begin{equation}\label{eq:Uwerlem}
 U=\begin{pmatrix}
    \frac{1}{\sqrt{2}} & \frac{1}{\sqrt{2}} &  \\
    \frac{i}{\sqrt{2}} & \frac{-i}{\sqrt{2}} &  \\
     & &  \mathds{1} 
   \end{pmatrix}
\end{equation}
that acts non-trivially only on the span of $\{\ket{1},\ket{2}\}$. Then
\[
 U\otimes U \ket{\psi_\walpha} = \sqrt{\frac{1-\walpha}{2}}(\ket{11}+\ket{22}) - i\sqrt{\frac{\walpha}{2}}(\ket{12}-\ket{21}).
\]
It holds that $\bra{\psi_\walpha}U^\dagger \otimes U^\dagger  Q U\otimes U \ket{\psi_\walpha} = 1-\walpha$ and thus $\calT_{\wer}^G(U\otimes U\psi_\walpha)=\rho_{\wer}^G(\walpha,1-\walpha)$, which completes the proof.
\end{proof}

\begin{proof}[Proof \textup{(}of Lemma \ref{lem:coorbwer} part 2\textup{)}] By convexity, it suffices to check only the states on the boundary, i.e.~$\rho_{\wer}^G(\walpha,\wbeta)$ with $\wbeta=0$ and $\wbeta=\frac{2(1-\walpha)}{d}$. In both cases, we will find a pure state $\ket{\psi}\in\orb_{\wer}(\psi_\walpha)$ such that $\calT_{\Orth}(\psi)=\rho_{\text{O}}^G(\walpha,\wbeta)$. Note that $\bra{\psi_\walpha}\Phi_d\ket{\psi_\walpha} = 0$. Thus $\calT_{\Orth}^G(\psi_\walpha)=\rho_{\text{O}}(a,0)$. With the same $U$ as in Eq.~\eqref{eq:Uwerlem}, it holds that $\bra{\psi_\walpha}U^\dagger \otimes U^\dagger  \Phi_d U\otimes U \ket{\psi_\walpha} = \frac{2(1-\walpha)}{d}$. This implies that $\calT_{\text{O}}(U\otimes U\psi_\walpha)=\rho_{\wer}^G(\walpha,\frac{2(1-\walpha)}{d})$ which completes the proof. 
\end{proof}

%-----------------------------------------------------------
\subsection{Proof of Lemma \ref{lem:coorbiso}}

Recall that, for any entanglement monotone $E$ and any $b\in[\frac{1}{d},1]$, the pure state that minimizes Eq.~\eqref{eq:Eisoequal} will be of the form
\begin{equation}\label{eq:phiiso}
\ket{\phi_b}=\sum_{i=1}^d\sqrt{\lambda_i}\ket{ii} 
\end{equation}
where the Schmidt coefficients satisfy $\sum_{i=1}^d\sqrt{\lambda}=\sqrt{db}$. 

\begin{proof}[Proof \textup{(}of Lemma \ref{lem:coorbiso} part 1\textup{)}]
 As above, it suffices to check only the states on the boundary. That is, we check $\rho_{\isot}^G(\walpha,\wbeta)$ with $\walpha=0$ and $\walpha=1-\wbeta$. In both cases, we will find a pure state $\ket{\psi}\in\orb_{\isot}(\phi_b)$ such that $\calT_{\isot}^G(\psi)=\rho_{\isot}^G(\walpha,\wbeta)$. Note that $\bra{\phi_{\wbeta}}Q\ket{\phi_{\wbeta}}=1$ and thus 
 \[
    \bra{\phi_{\wbeta}}(Q-\Phi_d)\ket{\phi_{\wbeta}}=1-\wbeta.
 \]
  Hence $\calT_{\isot}^G(\phi_{\wbeta})=\rho_{\isot}^G(1-\wbeta,\wbeta)$ and thus $\rho_{\isot}^G(1-\wbeta,\wbeta)\in\co(\orb_{\isot}(\phi_b))$. For $\rho_{\isot}^G(0,b)$, we use the discrete Fourier transform unitary matrix
  \[
    U=\frac{1}{\sqrt{d}}\sum_{j,k=1}^d \omega^{jk}\ket{j}\bra{k},
  \]
 where $\omega=e^{\frac{2i\pi}{d}}$ is the $d$\textsuperscript{th} root of unity. It holds that
 \begin{align*}
  \bra{\phi_b}U^{\dagger}\otimes\overline{U}{}&^{\dagger}QU\otimes\overline{U}\ket{\phi_b} 
     &= \frac{1}{d^2}\sum_{k=1}^d\Bigl(\sum_{j=1}^d\sqrt{\lambda_j}\abs{\omega^{jk}}^2\Bigr)^2\\ 
     &= \frac{1}{d}\Bigl(\sum_{j=1}^d\sqrt{\lambda_j}\Bigr)^2\\
     &=\wbeta.
 \end{align*}
 Thus $\bra{\phi_b}U^{\dagger}\otimes \overline{U}{}^{\dagger}(Q-\Phi_d)U\otimes\overline{U}\ket{\phi_{\wbeta}} =0$. This implies that $\calT_{\isot}^G(U\otimes\overline{U}\ket{\phi_{\wbeta}}) = \rho_{\isot}^G(0,\wbeta)$, which completes the proof.
\end{proof}

\begin{proof}[Proof \textup{(}of Lemma \ref{lem:coorbiso} part 2\textup{)}]
 It suffices to check only the states on the boundary. That is, we check $\rho_{\Orth}^G(\walpha,\wbeta)$ with $\walpha=0$ and $\walpha=\frac{d(1-b)}{2(d-1)}$. Note that $\braketinner{\phi_b}{W_{\!-}}{\phi_b}=0$ and thus $\calT_{\Orth}(\phi_b)=\rho_{\Orth}(0,b)$. Hence $\rho_{\Orth}(0,b)\in\co(\orb_{\isot}(\phi_b))$. 
 For $\rho_{\Orth}(\frac{d(1-b)}{2(d-1)},b)$, it suffices to find a unitary $U$ such that 
 \[
  \braketinner{\phi_b}{(U\otimes\overline{U})^\dagger W_{\!-}(U\otimes\overline{U})}{\phi_b} \geq \frac{d(1-b)}{2(d-1)}.
 \]
 We split the proof into two parts. First suppose that $\vec{\lambda}$ is of the form
 \[
  \vec{\lambda}=\bigl(t,\dots,t,\tfrac{1-kt}{d-k},\dots,\tfrac{1-kt}{d-k}\bigr)
 \]
 with $\ket{\phi_b}=\sum_{i=1}^d\sqrt{\lambda_i}\ket{ii}$ and 
 \begin{align*}
  \Bigl(\sum_{i=1}^d\sqrt{\lambda_i}\Bigr)^2 & = \bigl(k\sqrt{t}+ \sqrt{(d-k)(1-kt)}\bigr)^2 = db.
 \end{align*}
 For distinct indices $j,k\in\{1,2,\dots,d\}$ with $j<k$, define the unitary matrices
\begin{align*}
 U_{j,k} 
     &= \sum_{l\neq j,k} \ketbra{l}{l}+  \frac{1}{\sqrt{2}}\bigl(\ketbra{j}{j} +\ketbra{j}{k} + i\ketbra{k}{j} -i\ketbra{k}{k}\bigr)%\\
%      &=\begin{pmatrix}
%         1 &        &                    & &                     &  & \\
%           & \ddots &                    & &                     &  & \\
%           &        & \frac{1}{\sqrt{2}} & & \frac{1}{\sqrt{2}}  &  & \\
%           &        &                    & &                     &  &\\
%           &        & \frac{i}{\sqrt{2}} & & \frac{-i}{\sqrt{2}} &  & \\
%           &        &                    & &                     & \ddots&\\
%           &        &                    & &                     & &1
%        \end{pmatrix}
\end{align*}
that act nontrivially only on the subspace spanned by $\{\ket{j},\ket{k}\}$ and trivially elsewhere. Note that $U$ in Eq.~\eqref{eq:Uwerlem} is $U_{1,2}$ in this notation. Furthermore note that
\[
 \braketinner{\phi_b}{(U_{j,k}\otimes\overline{U_{j,k}})^\dagger  W_{\!-}(U_{j,k}\otimes\overline{U_{j,k}})}{\phi_b} = \frac{\left(\sqrt{\lambda_j}-\sqrt{\lambda_k}\right)^2}{2}.
\]
%and that $U_{j,k}U_{j',k'}=U_{j',k'}U_{j,k}$ if $j,k,j',k'$ are all distinct. 
Let $U=(U_{1,d})(U_{2,d-1})\cdots (U_{m,d+1-m})$, where $m=\min\{k,d-k\}$. Then
\begin{align*}
 &\bra{\phi_b}(U\otimes\overline{U})^\dagger W_{\!-} (U\otimes\overline{U})\ket{\phi_b} \\
  &\quad= \frac{(\sqrt{\lambda_1}-\sqrt{\lambda_d})^2}{2} +  \cdots + \frac{(\sqrt{\lambda_m}-\sqrt{\lambda_{d-m+1}})^2}{2}\\
  %&\quad= \frac{m}{2}\left(\sqrt{t}-\sqrt{\tfrac{1-kt}{d-k}}\right)^2\\
  &\quad= \frac{m}{2(d-k)}\left(\sqrt{(d-k)t}-\sqrt{1-kt}\right)^2\\
  %&\quad= \frac{m}{2(d-k)}\left((d-k)t + 1 - kt -2\sqrt{t(1-kt)(d-k)}\right)\\
  &\quad= \frac{m}{2(d-k)}\left(\frac{d(1-b)}{k}\right)\\
  &\quad=  \frac{d(1-b)}{2}\frac{\min\{k,d-k\}}{k(d-k)}\\
  %&\quad= \frac{d(1-b)}{2}\frac{1}{\max\{k,d-k\}}\\
  &\quad\geq \frac{d(1-b)}{2}\frac{1}{d-1}
\end{align*}
 with equality if and only if $k=d-1$ or $k=1$ (or $b=1$). The result follows. 
 
 The proof of the other case is analogous. In this case, suppose that $\vec{\lambda}$ is of the form 
 \[
  \vec{\lambda}=\bigl(t,\dots,t,1-kt,0,\dots,0\bigr)
 \]
 with $(\sum_{i=1}^d\sqrt{\lambda_i})^2  = (k\sqrt{t}+ \sqrt{1-kt})^2 = db$. Using the unitary $U=(U_{1,d})(U_{2,d-1})\cdots (U_{\lfloor\frac{d}{2}\rfloor,d-\lfloor\frac{d}{2}\rfloor+1})$, it is not difficult to show that 
 \[
  \bra{\phi_b}(U\otimes\overline{U})^\dagger W_{\!-} (U\otimes\overline{U})\ket{\phi_b}\geq \frac{d(1-b)}{2(d-1)}
 \]
 with equality if and only if $k=d-1$ (or $b=1$). 
 \end{proof}

%%%%%%%%%%%%%%%%%%%%%%%%%%%%%%%%%%%%%%%%%%%%%%%%%%%%%%%%%%%%%%%%

\section{Pure to isotropic conversion witness}
\label{sec:conversionwitnesscomputation}

For a fixed Schmidt vector $\vec{\lambda}$ we define
\[
 f_k(\vec{\mu}) = E_k(\vec{\lambda})-E_k(\vec{\mu}),
\]
and write this as $f_k(\vec{\mu}) = \mu_1+\cdots+\mu_k - (\lambda_1+\cdots+\lambda_k)$.
The goal is to compute
\[
 W_{\isot}(\vec{\lambda},\vec{\mu})=\max_{\vec{\mu}}\min_k  f_k(\vec{\mu}).
\]
We can split this into $d-1$ separate optimization problems as follows. For each $k\in\{1,\dots,d-1\}$, we maximize $f_k(\vec{\mu})$ over all $\vec{\mu}$ for which $k$ yields the minimum. That is, maximize over all $\vec{\mu}$ for which
$f_k(\vec{\mu})\leq f_\ell(\vec{\mu})$ for all $\ell\in\{1,\dots,d-1\}$. Minimizing this over all $k$ yields the desired result
\[
 W_{\isot}(\vec{\lambda},\vec{\mu}) = \min_k \left[\max_{\vec{\mu}}\left\{f_k(\vec{\mu})\middle|f_k(\vec{\mu})\leq f_\ell(\vec{\mu})\text{ for all } \ell\right\} \right]
\]
where the maximizations are taken over all Schmidt vectors satisfying $\sum_{i=1} ^d\sqrt{\mu_i}=\sqrt{d\wbeta}$.

For each $k$, these suboptimization problems can be rewritten as follows:
\begin{align*}
 \text{maximize: }  \, &\sum_{i=1}^k\mu_i\\
 \text{subject to: }\, &\sum_{i=1}^d\mu_i = 1 \\
                       &\sum_{i=1}^d\sqrt{\mu_i} = \sqrt{d\wbeta} \\
                       &\sum_{i=2}^\ell\mu_i\leq \sum_{i=2}^\ell\lambda_i\text{ for all }\ell\in\{1,\dots,k-1\}\\
                       &\sum_{i=k+1}^{\ell+1}\lambda_i\leq \sum_{i=k+1}^{\ell+1}\mu_i\text{ for all }\ell\in\{k+1,\dots,d-1\}.
\end{align*}
There are $d$ constraints for these $d$-dimensional optimization problems, so we may use the method of Lagrange multipliers to find optimal solutions. 

%%%%%%%%%%%%%%%%%%%%%%%%%%%%%%%%%%%%%%%%%%%%%%%%%%%%%%%%%%%%%%%%

\end{document}